\documentclass[journal]{IEEEtran}
\usepackage[utf8]{inputenc}
\usepackage[dvipsnames]{xcolor}
\usepackage{color}
\usepackage{graphics}
\usepackage{graphicx}
\usepackage{epsfig,epstopdf}
\usepackage{bm,bbm}
\usepackage{mathrsfs}
\usepackage{pstricks,psfrag}
\usepackage{pst-plot}
\usepackage{amssymb,amsmath,amsthm,amscd,amsfonts,amsbsy}
\usepackage{lettrine}
\usepackage{authblk}
\usepackage{float}
\usepackage{cite}
\usepackage{balance}
\usepackage{dsfont}
\usepackage{setspace}
\usepackage{relsize}
\usepackage{comment}
\usepackage{algorithm}
\usepackage[noend]{algpseudocode}
\usepackage[nolist]{acronym}
\usepackage{arydshln}
\usepackage{pifont}
\usepackage{units}
\usepackage{soul}
\usepackage{tikz}
\usetikzlibrary{
	calc,
	shapes,
	arrows,
	shapes.misc,
	shapes.arrows,
	chains,
	matrix,
	positioning,
	scopes,
	decorations.pathmorphing,
	shapes.geometric,
	shadows,
	decorations.pathreplacing,
	patterns
}
\usepackage{tikzscale}
\usepackage{pgfplots}
\pgfplotsset{compat=newest}
\pgfplotsset{plot coordinates/math parser=false} 
\newlength\figureheight 
\newlength\figurewidth 
\usepackage{subcaption}
\newtheorem{lemma}{Lemma}

\newtheorem{remark}{Remark}
\newcommand{\el}{(\mathrm{el})}
\newcommand{\az}{(\mathrm{az})}
\newcommand{\BS}{\mathrm{BS}}
\newcommand{\UE}{\mathrm{UE}}

\newcommand{\A}{\mathrm{A}}
\newcommand{\D}{\mathrm{D}}
\newcommand{\R}{\mathbf{R}}
\newcommand{\vecR}{\mathbf{r}}
\newcommand{\p}{\mathbf{p}}
\newcommand{\Transpose}{^\mathsf{T}}
\newcommand{\Hermitian}{^\mathsf{H}}
\newcommand{\tr}{\mathrm{tr}}
\newcommand{\vect}{\mathrm{vec}}

\newcommand{\CCRB}{^{(\mathrm{CCRB})}}
\newcommand{\acos}{\mathrm{acos}}
\newcommand{\atan}{\mathrm{atan}}
\makeatletter
\newcommand{\gettikzxy}[3]{%
  \tikz@scan@one@point\pgfutil@firstofone#1\relax
  \edef#2{\the\pgf@x}%
  \edef#3{\the\pgf@y}%
}
\newcommand{\cmark}{\ding{51}}%
\newcommand{\xmark}{\ding{55}}%

\newcommand{\showfigure}[1]{#1} % use this for revealing the solution 
\begin{acronym}[ACRONYM]
\acro{mmWave}{millimeter-wave}
\acro{MIMO}{multiple-input multiple-output}
\acro{OFDM}{orthogonal frequency-division multiplexing}
\acro{MIMO-OFDM}{multiple-input multiple-output orthogonal frequency-division multiplexing}
\acro{AoA}{angle-of-arrival}
\acro{AoD}{angle-of-departure}
\acro{ToA}{time-of-arrival}
\acro{TDoA}{time-difference-of-arrival}
\acro{CRB}{Cram\'er-Rao bound}
\acro{CCRB}{constrained Cram\'er-Rao bound}
\acro{FIM}{Fisher information matrix}
\acro{EFIM}{equivalent Fisher information matrix}
\acro{RMSE}{root mean square error}
\acro{ML}{maximum likelihood}
\acro{GNSS}{global navigation satellite system}
\acro{IMU}{inertial measurement unit}
\acro{PRS}{positioning reference signal}
\acro{QOS}{quality of service}
\acro{AWGN}{additive white Gaussian noise}
\acro{PSD}{power spectral density}
\acro{RF}{radio frequency}
\acro{RIS}{reconfigurable intelligent surface}
\acro{RTT}{round-trip-time}
\acro{SLAM}{simultaneous localization and mapping}
\acro{SNR}{signal-to-noise ratio}
\acro{BS}{base station}
\acro{UE}{user equipment}
\acro{IP}{incidence point}
\acro{LoS}{line-of-sight}
\acro{NLoS}{non-line-of-sight}
\acro{TXRX}[TX/RX]{transmitter/receiver}
\acro{TX}{transmitter}
\acro{RX}{receiver}
\acro{PEB}{position error bound}
\acro{MEB}{mapping error bound}
\acro{OEB}{orientation error bound}
\acro{IPEB}{IPs error bound}
\acro{SEB}{synchronization error bound}
\acro{CDF}{cumulative distribution function}
\acro{PDF}{probability distribution function}
\acro{UPA}{uniform planar array}
\acro{RBL}{rigid body localization}
\acro{VLP}{visible light positioning}
\acro{RSS}{received signal strength}
\acro{GLRT}{generalized likelihood ratio test}
\end{acronym}
\begin{document}
\bstctlcite{IEEEexample:BSTcontrol}
\title{MmWave 6D Radio Localization with a Snapshot Observation from a Single BS}
% \author{
% Mohammad A. Nazari\orcidlink{0000-0002-8786-1510},~\IEEEmembership{Student Member,~IEEE,} 
% Gonzalo Seco-Granados\orcidlink{0000-0003-2494-6872},~\IEEEmembership{Senior Member,~IEEE,} 
% Pontus Johannisson\orcidlink{0000-0003-1793-4451}, 
% and Henk Wymeersch\orcidlink{0000-0002-1298-6159},~\IEEEmembership{Senior Member,~IEEE}
\author{
Mohammad A. Nazari,~\IEEEmembership{Student Member,~IEEE,} 
Gonzalo Seco-Granados,~\IEEEmembership{Senior Member,~IEEE,} 
Pontus Johannisson, 
and Henk Wymeersch,~\IEEEmembership{Senior Member,~IEEE}
\thanks{
This work was supported under the Wallenberg AI, Autonomous Systems and Software Program (WASP), the Swedish Research Council under grant 2018-03701, the European Commission through the H2020 project Hexa-X (grant agreement no. 101015956), the Spanish Research Agency project PID2020-118984GB-I00, and by the ICREA Academia Programme.
\\
Mohammad A. Nazari and Henk Wymeersch are with the Department
of Electrical Engineering, Chalmers University of Technology, 41258 G\"oteborg, Sweden (e-mails: \{mohammad.nazari,henkw\}@chalmers.se).}
\thanks{Gonzalo Seco-Granados is with the Department of Telecommunications and Systems Engineering, Universitat Autonoma de Barcelona, 08193 Bellaterra, Barcelona, Spain (e-mail: gonzalo.seco@uab.cat).}
\thanks{Pontus Johannisson was with Chalmers University of Technology, and is now with Saab AB, 41276 G\"oteborg, Sweden (e-mail: pontus.johannisson@saabgroup.com).}}
% \markboth{IEEE TRANSACTIONS ON WIRELESS COMMUNICATIONS}%
% {Nazari \MakeLowercase{\textit{et al.}}: MmWave 6D Radio Localization with a Snapshot Observation from a Single BS}
\newcommand{\ColorBlue}[1]{\textcolor{black}{#1}}
\maketitle
\begin{abstract}
Accurate and ubiquitous localization is crucial for a variety of applications such as logistics, navigation, intelligent transport, monitoring, control, {and also for the benefit of communications}. Exploiting {\ac{mmWave}} signals in 5G and Beyond 5G systems can provide accurate localization with limited infrastructure. We consider the single %base station 
{\ac{BS}} localization problem and extend it to 3D position and 3D orientation estimation of an unsynchronized multi-antenna %user, 
{\ac{UE}}, using downlink {\ac{MIMO-OFDM}} signals. Through a Fisher information analysis, we show that the problem is often identifiable, provided that there is at least one %additional 
multipath component in addition to the {\ac{LoS}}, even if the position of corresponding {\ac{IP}} is a priori unknown. Subsequently, we pose a {\ac{ML}} estimation problem, to jointly estimate the 3D position and 3D orientation of the \ac{UE} as well as several nuisance parameters (the \ac{UE} clock offset and the positions of \acp{IP} corresponding to the multipath). The \ac{ML} problem is a high-dimensional non-convex optimization problem over a product of Euclidean and non-Euclidean manifolds. To avoid complex exhaustive search procedures, we propose a geometric initial estimate of all parameters, which reduces the problem to a 1-dimensional search over a finite interval. Numerical results show the efficiency of the proposed ad-hoc estimation, whose gap to the \ac{CRB} is tightened using the \ac{ML} estimation.
\end{abstract}
\acresetall
\section{Introduction} 
\Ac{mmWave}
is the key enabling component of the fifth generation (5G) and beyond 5G (B5G) communication systems, which empowers the implementation of large antenna arrays for spatial multiplexing and provides massive bandwidths for high data rates \cite{rappaport2013millimeter}. Despite the favorable properties of \ac{mmWave}, undesired effects, such as severe path loss and limited channel rank, challenge the technology to come up with advanced beamforming and resource allocation schemes \cite{roh2014millimeter}. To achieve this, information on the 3D location of mobile users can provide important side information, so that the \ac{BS} can adjust its precoders to beam towards the \ac{UE} \cite{di2014location}. Similarly, the \ac{UE} can adjust its combiners, based on its 3D orientation, to maximize the received \ac{SNR} \cite{alammouri2019hand}. 
\ColorBlue{While \ac{UE} position estimation has been the main driver in 5G mmWave \cite{shastri2022review}, there are many applications that need 6D information (3D position and 3D orientation, also known as the pose in robotics \cite{thrun2002probabilistic}):  the position and heading of vehicles is needed in intelligent transport systems for driving assistance applications and platooning \cite{bartoletti2021positioning}; in assisted living facilities, the pose of residents is informative about their health status \cite{HighAccuracyLocalizationforAssistedLiving--Witrisal_others}; search-and-rescue operations involving UAVs require accurate and timely pose information for control, self-localization, and victim recovery \cite{Albanese2021responders}. Moreover, 6D localization is expected to be of importance in 6G, with applications such as augmented reality, robot interaction, and digital twins \cite{behravan2022positioning}.}

The source of this 6D information, whose estimation is referred to as \emph{6D localization} in this paper, can be either external or internal to the communication systems. The external 6D localization systems can build on %wide variety of mixed-technology solutions \cite{IndoorTracking--D.Dardari_others}, 
a mixed-technology solution, such as the combination of the \ac{GNSS} (for 3D position) and \ac{IMU} (for 3D orientation) \cite{wahlstrom2017smartphone}. However, such solutions can be inefficient in cost, complexity, or coverage. For example, \ac{GNSS} might fail in indoor environments or urban canyons, while \acp{IMU} suffer from drifts and accumulative errors \cite{IndoorTracking--D.Dardari_others}. The alternative is to exploit the already deployed cellular communication infrastructure for 6D localization, and feed the 6D information to the communication system internally.

Prior to 5G, the majority of localization schemes in cellular networks relied on multiple synchronized base stations and \ac{TDoA} measurements \cite{SurveyCellularRadioLocalization--Rosado_others_G.Seco-Granados}. With the introduction of new dedicated \acp{PRS} and measurements in 3GPP release 16\cite{tr:38855-3gpp19}, a combination of angle and delay measurements has become possible \cite{Positioning5GNetworks--Dwivedi_others}. Because of the high resolution in both temporal and angular domains, thanks to $400$ MHz bandwidth at \ac{mmWave} and large antenna arrays, respectively, multipath components can be better resolved \cite{ErrorBoundsfor3DLocalization--Z.Abu-Shaban_others_G.Seco-Granados_H.Wymeersch}, leading to new positioning architectures. A great deal of research effort has been devoted to \emph{multi-\ac{BS}\footnote{Recent developments relying on \acp{RIS} are considered as multi-BS solutions, since a \ac{RIS} acts as an additional multi-antenna \ac{BS} in localization \cite{bjornson2021reconfigurable,RISforLocalization--A.Elzanaty_others_M.S.Alouini}. For similar reasons, approaches with a single moving \ac{BS} are also equivalent to multi-\ac{BS} localization \cite{Albanese2021responders}.} localization} approaches, which exploit these novel features \cite{DirectLocforMassiveMIMO--N.Garcia_H.Wymeersch_R.Larsson_others,PositionLocforFuturisticCommunications--O.Kanhere_T.Rappaport,JointCommAndLocMmWave--G.Kwon_others_M.Win,3DOrientationEstimation--M.Nazari_G.Seco-Granados_P.Johannisson_H.Wymeersch}. Simultaneously, there has been a paradigm shift towards \emph{single-\ac{BS} localization} solutions  \cite{SingleAnchorLocAndSynchFullDuplexAgents--Y.Liu_Y.Shen_M.Z.Win}, which are attractive because they require only minimal infrastructure and remove the need for inter-\ac{BS} synchronization. 
The enabler of single-\ac{BS} localization is the ability to turn multipath from foe to friend \cite{HighAccuracyLocalizationforAssistedLiving--Witrisal_others}: in contrast to previous beliefs that \ac{NLoS} components have unfavorable effects on positioning, they contribute to the identifiability and accuracy of localization in \ac{mmWave} {\ac{MIMO}} systems, provided there is sufficient temporal and spatial resolution. 
%This concept has been exploited both in tracking and in snapshot-based localization. In tracking, the \ac{UE} location and orientation are inferred over time, combining measurements with knowledge of the \ac{UE} dynamics in Bayesian filtering \cite{HighAccuracyLocalizationforAssistedLiving--Witrisal_others,leitinger2019belief,kim20205g}. In snapshot-based localization, the entire 6D state vector is inferred from the signals received during a short time interval, thereby avoiding the need of modeling the time evolution of the parameters. This renders the corresponding positioning methods similar to all 5G snapshot-based positioning techniques and thus easy to integrate into standards.
This concept has been exploited in recent advances in localization. In particular, \cite{PositionOrientationEstimation--A.Shahmansoori_others_G.Seco-Granados_H.Wymeersch} derived both performance bounds and a method for 2D position and 1D orientation estimation, with synchronized \ac{UE} and \ac{BS}, based on a compressed sensing algorithm. The estimation method was refined via an atomic norm minimization approach in \cite{li2022joint}, where the performance is not limited by quantization error and grid resolution. For the same scenario of 2D position and 1D orientation, \cite{NLoSforPositionOrientationEstimation--R.Mendrzik_H.Wymeersch_G.Bauch} showed that each \ac{NLoS} path gives rise to a rank-1 \ac{FIM} so that the \ac{UE} can be localized with the \ac{LoS} and a single \ac{NLoS} path, or with 3 \ac{NLoS} paths when the \ac{LoS} is obstructed. The case of obstructed \ac{LoS} was also treated in \cite{wen20205g}, without the requirement for synchronization, {but still with 2D position and accordingly a single orientation angle}. This concept was further extended to 3D position and 2D orientation estimation under perfect synchronization in \cite{ErrorBoundsfor3DLocalization--Z.Abu-Shaban_others_G.Seco-Granados_H.Wymeersch}, where the asymptotic case with orthogonal multipath components was studied. A further generalization was considered in \cite{SingleAnchorLocalizationLimits--A.Guerra_others_D.Dardari}, focusing on a direct localization approach for the massive array regime, considering a 3D position and 2D orientation estimation and a synchronized user. The more practical case considering the synchronization error as well as the Doppler shift-when the transmitter or receiver is moving-was addressed in \cite{PerformanceLimitsSingleAnchorMillimeterWavePositioning--A.Kakkavas_others_J.A.Nossek}, where the authors performed the \ac{FIM} analysis for 2D position and 1D orientation estimation of a mobile user. 
\ColorBlue{In \cite{s21041178}, the \ac{LoS} 3D positioning problem using the 5G uplink channel sounding reference signals is considered, but the \ac{UE} is synchronized and single-antenna, for which no orientation is defined. The closest work  we identified is \cite{Nuria}, where a hybrid model/data-driven approach for the 3D position estimation of a \ac{UE} is proposed. The model-based approach is based on the underlying geometry and the angle between each two arrival directions being independent of \ac{UE} orientation. However, no algorithm is proposed for orientation estimation or estimating the \ac{IP} locations.}

%The more practical case of an unsynchronized user was addressed in \cite{PerformanceLimitsSingleAnchorMillimeterWavePositioning--A.Kakkavas_others_J.A.Nossek}, which performed the \ac{FIM} analysis for 2D position and 1D orientation estimation of a mobile user, considering both the synchronization error as well as the Doppler shift, when the transmitter or receiver is moving.

\ColorBlue{The extension of already existing methods to 3D position and 3D orientation case, however, is not trivial,  since the positions (both \ac{UE} and \acp{IP}) are not constrained to lie in a plane. The orientation also introduces more degrees of freedom in the geometrical equations of channel parameters, and hence, increases the complexity of search-based estimation algorithms. In addition, the analysis of fundamental lower bounds for such a general case has not yet been conducted.}
The 6D localization problem was introduced in the \ac{mmWave} context in \cite{MassiveMIMOIsReality}, but has not yet been further developed. However, it has been studied in other settings, e.g., pose estimation in robotics \cite{thrun2002probabilistic} and visible light positioning \cite{SPOforVLC--Shen_others_Steendam}. In \cite{SPOforVLC--Shen_others_Steendam}, a simultaneous 3D position and 3D orientation estimation using the \ac{RSS} for a visible light system containing multiple light emitting diodes and photodiodes is considered, and an approximate solution using direct linear transformation method is proposed. This solution is further refined using iterative algorithms for ML estimation. On a parallel track, the problem is addressed under the label of rigid body localization in \cite{RBL_using_ArrivalTime_and_Doppler--J.Jiang_others_K.C.Ho,RBL_using_SensorNetworks--S.Chepuri_others_A.Veen,AngleBasedRBL--Y.Wang_others_K.C.Ho_L.Huang}, where the approach is to mount sensors with a known topology on the body. The positions of the sensors in the global coordinate frame are related to the position of the rigid body and its orientation. The sensors then form a wireless sensor network, and the position, as well as the orientation of the body, is estimated using time and/or angle measurements from sensors.
%However, it has been studied in other settings, e.g., visible light positioning \cite{SPOforVLC--Shen_others_Steendam} and rigid body localization \cite{RBL_using_ArrivalTime_and_Doppler--J.Jiang_others_K.C.Ho,RBL_using_SensorNetworks--S.Chepuri_others_A.Veen,AngleBasedRBL--Y.Wang_others_K.C.Ho_L.Huang}, and of course in robotics \cite{thrun2002probabilistic}. In \cite{SPOforVLC--Shen_others_Steendam}, a simultaneous 3D position and 3D orientation estimation using the \ac{RSS} for a visible light system containing multiple LEDs and multiple photodiodes is considered, and an approximate solution using direct linear transformation method is proposed. This solution is further refined using iterative algorithms for ML estimation. On a parallel track, the problem is addressed  under the label of \ac{RBL} in \cite{RBL_using_ArrivalTime_and_Doppler--J.Jiang_others_K.C.Ho,RBL_using_SensorNetworks--S.Chepuri_others_A.Veen,AngleBasedRBL--Y.Wang_others_K.C.Ho_L.Huang}. Their approach is to mount sensors with a known topology on the body. The positions of the sensors in the global coordinate frame are related to the position of the rigid body and its orientation. The sensors then form a wireless sensor network, and the position as well as the orientation of the body is estimated using time and/or angle measurements from sensors.

In this paper, we consider a single-\ac{BS} localization scenario, where the downlink \ac{mmWave} signal from a multi-antenna \ac{mmWave} base station is used to estimate the 3D-position and 3D-orientation of a \ac{UE} in \ac{LoS} to the \ac{BS}. We evaluate the lower bound on estimation error variance of position and orientation of the \ac{UE}, positions of the incidence points, and the clock offset, by deriving the constrained \acf{CRB} of all unknowns. This reveals that the problem is generally identifiable under a single \ac{NLoS} path. To solve the corresponding high-dimensional \ac{ML} problem, we propose an efficient solution, combining a geometric ad-hoc estimator to initialize a gradient descent over a product of manifolds. This solution is shown to attain the corresponding \acp{CRB}. The proposed approach is related to the literature in Table \ref{tab:related-work}.
The main contributions of this work are the following:
\ColorBlue{
\begin{itemize}
\item \textbf{6D localization algorithm:} We pose a high-dimensional \ac{ML} estimation problem over a product of Euclidean and non-Euclidean manifolds, given the conditional probability distributions of \ac{AoA}, \ac{AoD}, and \ac{ToA} measurements. The parameters of such distributions are obtained, and the \ac{RMSE} of the \ac{ML} estimation is shown to attain the lower bounds. The \ac{ML} problem is solved by gradient descent, iterating between the various manifolds, starting from a good initial solution. 
\item \textbf{A low-complexity estimator} We propose and evaluate a low-complexity ad-hoc estimation algorithm to initialize the solution of the \ac{ML} problem, which reduces the high-dimensional problem of estimating all unknowns, to a 1-dimensional search over a finite interval combined with closed-form expressions. This recovers not only the 6D \ac{UE} state, but also the \ac{UE} clock bias and \ac{IP} locations. 
\item \textbf{Fisher information and numerical analysis:} We obtain the lower bound on the estimation of the 6D user state and its clock bias, as well as the map of the environment, i.e., the positions of incidence points, and then we evaluate the impact of bandwidth, number of antennas, and number of incidence points. The analysis of the bounds indicates that in most cases a single incidence point with a priori unknown location is sufficient to render the problem identifiable, though certain configurations require several incidence points. We also evaluate a low-complexity ad-hoc  and the ML 6D localization algorithms and demonstrate that it can attain the corresponding performance bounds.
\end{itemize}}

\begin{table}%[t]
    \centering
    \caption{Overview of the related work.}
    \label{tab:related-work}
    % \rowcolors{3}{gray!10}{white}
    \resizebox{\columnwidth}{!} {
    \begin{tabular}{|c|c|c|c|c|c|c|}
        \hline 
        \textbf{Ref.}  & \textbf{Pos.}  & \textbf{Ori.}  & \textbf{Clock} & \textbf{Bound}  & \textbf{Method} & \textbf{without \ac{LoS}}   \tabularnewline
        \hline 
        \cite{PositionOrientationEstimation--A.Shahmansoori_others_G.Seco-Granados_H.Wymeersch} & 2D & 1D & \xmark & \cmark & \cmark & \cmark\\
        \hline 
        \cite{li2022joint} & 2D & 1D & \xmark & \cmark & \cmark & \xmark\\
        \hline 
        \cite{NLoSforPositionOrientationEstimation--R.Mendrzik_H.Wymeersch_G.Bauch} & 2D & 1D & \xmark & \cmark & \xmark & \cmark\\
        \hline 
        \cite{wen20205g} & 2D & 1D & \cmark & \xmark & \cmark & \cmark\\
        \hline 
        \cite{ErrorBoundsfor3DLocalization--Z.Abu-Shaban_others_G.Seco-Granados_H.Wymeersch} & 3D & 2D & \xmark & \cmark & \xmark & \xmark\\
        \hline 
        \cite{SingleAnchorLocalizationLimits--A.Guerra_others_D.Dardari} & 3D & 2D & \xmark & \cmark & \xmark & \xmark\\
        \hline 
        \cite{PerformanceLimitsSingleAnchorMillimeterWavePositioning--A.Kakkavas_others_J.A.Nossek} & 2D & 1D & \cmark & \cmark & \cmark & \xmark\\
        \hline 
        \cite{s21041178} & 3D & \xmark & \xmark & \xmark & \cmark & \xmark\\
        \hline 
        \cite{Nuria} & 3D & 3D & \cmark & \xmark & only pos. & \xmark\\
        \hline
        this work & 3D & 3D & \cmark & \cmark & \cmark & \xmark\\
        \hline 
    \end{tabular}}\vspace{-5mm}
\end{table}

The rest of the paper is organized as follows. Section \ref{sec:system model} describes the system model and provides the definitions. In section \ref{sec:Methodology}, we state the ML estimation problem, followed by a low-complexity estimation algorithm for obtaining initial solutions for iterative routines in Section \ref{sec:Ad-hoc Initial Estimate}. Subsequently, in Section \ref{sec:FIM Analysis}, the Fisher information analysis is done, and error bounds are derived. Section \ref{section:result} presents numerical results, and finally, conclusion remarks are given in Section \ref{section:conclusion}.  

\subsubsection*{Notations} We denote scalars, vectors, and matrices by italic, bold lowercase, and bold uppercase letters, e.g., $x$, $\mathbf{x}$ and $\mathbf{X}$, respectively. The element of matrix $\mathbf{X}$ in the $i$-th row and $j$-th column is indicated by $[\mathbf{X}]_{i,j}$. We also use $[\mathbf{X}]_{i:k,j:l}$ to refer to the sub-matrix lying between rows $i$ to $k$ and columns $j$ to $l$ of $\mathbf{X}$. The identity matrix of size $N$ is shown by $\mathbf{I}_N$, whereas $\mathbf{1}_N$ and $\mathbf{0}_N$ indicate all-ones and all-zeros vectors of size $N$. While $\mathrm{diag}(\mathbf{x})$ is a diagonal matrix whose non-zero elements are given by $\mathbf{x}$, $\mathrm{diag}(\mathbf{X})$ is a vector composed of the diagonal elements of $\mathbf{X}$. Similarly, $\mathrm{blkdiag}(\mathbf{X},\mathbf{Y})$ is a block-diagonal matrix made of $\mathbf{X},\mathbf{Y}$. In order to show the expectation, trace, and vectorization operators, we use $\mathbb{E}[\cdot]$, $\tr[\cdot]$, and $\vect[\cdot]$, respectively. The transpose, and hermitian operators are symbolized using $[\cdot]\Transpose$ and $[\cdot]\Hermitian$; and we consider $\odot$ and $\oslash$ in conjunction with pointwise product and division, all in the given order. The Euclidean and Frobenius norms are denoted by $\Vert\cdot\Vert$ and $\Vert\cdot\Vert_{\text{F}}$, respectively.%, and finally, $[\hat{.}]$ is reserved for the measured or estimated parameters.

\section{System Model} \label{sec:system model}
We consider a downlink \ac{mmWave} \ac{MIMO} scenario consisting of a BS equipped with an arbitrary array of $N_{\BS}$ antennas, and a UE equipped with an arbitrary array of $N_{\UE}$ antennas, as shown in Fig.~\ref{fig:schematic of system model with LoS and NLoS paths}. Without loss of generality, both \ac{BS} and \ac{UE} are considered to have a single \ac{RF} chain.
%{\ColorBlue{\footnote{\ColorBlue{We consider precoding and combining vectors to allow for single \ac{RF} chain and multiple antennas. In case of multiple \ac{RF} chains, precoding and combining matrices are used.}}}. 

\subsection{Geometric Model} \label{subsec:geometric model}
The BS antenna array is centered at the {known} position $\p_{\BS}=[p_{\BS,x},p_{\BS,y},p_{\BS,z}]\Transpose\in\mathbb{R}^3$ with a {known} orientation, while UE antenna array is centered at the \emph{unknown} position $\p_{\UE}=[p_{\UE,x},p_{\UE,y},p_{\UE,z}]\Transpose\in\mathbb{R}^3$ with an \emph{unknown} orientation. The paths between BS and UE include the \ac{LoS}, as well as $M \ge 1$ \emph{resolvable single-bounce} \ac{NLoS} paths,%\footnote{\ColorBlue{In practice, the  multi-bounce paths may also exist. We assume that those paths are identified and excluded from the localization scheme. This can be done using methods from the literature, such as  \cite{li2022iterative}, where an iterative method based on the \ac{GLRT} and change-point detection is proposed to distinguish single-bounce and multi-bounce paths.}}
each corresponding to a scatterer or a reflecting point, represented by an \ac{IP} at \emph{unknown} position $\p_m=[p_{m,x},p_{m,y},p_{m,z}]\Transpose\in\mathbb{R}^3$, $m=1,\ldots,M$. All positions are given in a global coordinate frame as the reference, whose axes are labeled as $x$, $y$, and $z$ (see Fig.~\ref{fig: global coordinate system and reference orientation}). %The rotations are also expressed with respect to the same reference. 
% We use the terms reference coordinate frame and global coordinate frame interchangeably.
\ColorBlue{\begin{remark}[On the \ac{NLoS} model]
While we assume single-bounce \ac{NLoS} paths, in practice   multi-bounce paths may also exist. Those paths can be identified and excluded from the localization scheme, e.g., using methods such as  \cite{li2022iterative} (where an iterative method based on the \ac{GLRT} and change-point detection is proposed to distinguish single-bounce and multi-bounce paths) or \cite{hong2022joint} (which progressively identifies the \ac{LoS} path, single bounce, double bounce, and higher bounce paths). We also assume that all paths (including the \ac{LoS} and  \ac{NLoS} paths) are resolvable. This is generally correct in the mmWave scenarios, due to the large bandwidth and antenna apertures. 
\end{remark}}

The orientation of BS and UE describe how their antenna arrays with respect to the reference orientation are arranged. 
% In particular, we define a reference orientation for \ac{BS} and \ac{UE} arrays according to the axes of the global coordinate frame, where the array is parallel to the global XY-axes, as shown in Fig.~\ref{fig: global coordinate system and reference orientation}. 
We consider rotation matrices for describing the orientations\footnote{Orientations in 3D can be represented by \emph{quaternions}, \emph{Euler angles}, and \emph{rotation matrices} \cite{TutorialOnSE(3)--Blanco}. We adopt rotation matrices.}, so that BS and UE orientations determine local coordinate frames, respectively described by $3 \times 3$ rotation matrices $\R_\BS$ and $\R_\UE$, in the special orthogonal group $\mathrm{SO}(3)$. 
In particular, we define a reference orientation where the axes are in the same direction as those of the global coordinate frame, as shown in Fig.~\ref{fig: global coordinate system and reference orientation}. 
The local coordinate frames are thus obtained by rotating the arrays in reference orientation through $\R_\BS$ and $\R_\UE$ (see Fig.~\ref{fig: Local coordinate system and definition of angles}). Therefore, the given vector $\mathbf{y}$ in the global coordinate system is corresponding to $\mathbf{y}_\BS = \R_\BS\Transpose \mathbf{y}$ and $\mathbf{y}_\UE = \R_\UE\Transpose \mathbf{y}$ in \ac{BS} and \ac{UE} coordinate frames, respectively. We note that $\R_\BS$ is known, while $\R_\UE$ is unknown. The 6D localization problem refers to estimation of $\p_{\UE}$ and $\R_\UE$. 

\begin{figure}[t!]
 \centering
    \begin{tikzpicture}
    \node(image)[anchor=south west] {\includegraphics[height=3.3cm]{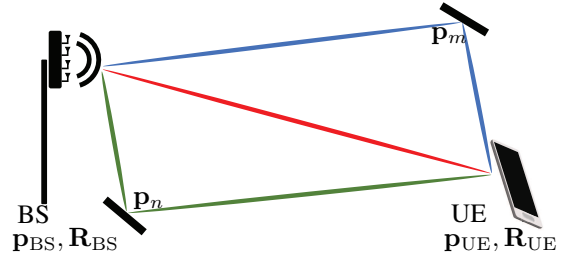}};
    \gettikzxy{(image.north east)}{\ix}{\iy};
    % \draw[help lines] (0,0) grid (\ix,\iy);
    %************************************
    \node at (6.7/8.4*\ix,3.4/4.2*\iy){$\p_m$};
    \node at (2/8.4*\ix,0.8/4.2*\iy){$\p_n$};
    \node at (0.2/8.4*\ix,0.6/4.2*\iy){\ac{BS}};
    \node at (0.7/8.4*\ix,0.2/4.2*\iy){$\p_\BS,\R_\BS$};
    \node at (7/8.4*\ix,0.6/4.2*\iy){\ac{UE}};
    \node at (7.5/8.4*\ix,0.2/4.2*\iy){$\p_\UE,\R_\UE$};
    %************************************
    \end{tikzpicture}
    \caption{Schematic of system model with a UE at unknown 3D position and unknown 3D orientation, where signals are received from \ac{LoS} and \ac{NLoS} paths. The UE and BS are not synchronized.}
    \label{fig:schematic of system model with LoS and NLoS paths}
    \vspace{-5mm}
\end{figure}
\vspace{-4mm}
\subsection{Signal and Channel Model}
%\vspace{-2mm}
We consider downlink pilot transmission of $K$ \ac{OFDM} symbols using $N_f$ subcarriers over a \ac{MIMO} channel, for the purpose of snapshot position and orientation estimation\footnote{We assume a slow fading channel, which does not vary over the duration of the pilot transmission.}. Considering $M+1$ (resolvable) paths indexed by $m=0,1,\ldots,M$ ($m=0$ is the \ac{LoS} path between the \ac{BS} and the \ac{UE}, which is assumed to be present, and $m\neq 0$ correspond to \ac{NLoS} paths), the $k$-th received \ac{OFDM} symbol, $k=1,\ldots,K$, over subcarrier $n=1,\ldots,N_f$, is given by \ColorBlue{\cite{huang2018iterative}}
\begin{align} \label{eq: signal model}
    y_{k,n} = \bm{w}_k\Hermitian \sum_{m=0}^{M} \mathbf{H}_{m} \bm{f}_k x_{k,n}\mathrm{e}^{-j2\pi (n-1) \Delta_f \tau_m} + \bm{w}_k\Hermitian\bm{n}_{k,n},
\end{align}
in which $x_{k,n}$ is a unit-modulus pilot symbol with energy $E_\mathrm{s}=P_{\mathrm{TX}}/B$, where $P_{\mathrm{TX}}$ is the transmit power and $B=N_f \Delta_f$ is the total bandwidth for the subcarrier spacing $\Delta_f$; $\bm{f}_k$ is the precoding vector at \ac{BS}, $\bm{w}_k$ is the combining vector at \ac{UE}, with $\Vert\bm{f}_k\Vert=\Vert\bm{w}_k\Vert=1$, $\forall k$; and $\bm{n}_{k,n} \in \mathbb{C}^{N_\UE}$ is the complex zero-mean \ac{AWGN} with covariance matrix $n_0 N_0 \mathbf{I}_{N_\UE}$, where $N_0$ is the noise \ac{PSD}, and $n_0$ is the \ac{UE} noise figure. 
Furthermore, each path $m$ is characterized by a \ac{ToA} $\tau_m$ and a channel matrix $ \mathbf{H}_m \triangleq h_{m}~ \mathbf{a}_{\UE}(\bm{\theta}_{\A,m}) \mathbf{a}_{\BS}\Transpose(\bm{\theta}_{\D,m})$, where $h_{m}=h_{\mathrm{R},m}+jh_{\mathrm{I},m}$ is the complex channel gain; $\mathbf{a}_{\UE}(\bm{\theta}_{\A,m})$ and $\mathbf{a}_{\BS}(\bm{\theta}_{\D,m})$ are the array response vectors, where \smash{$\bm{\theta}_{\A,m} = [\theta_{\A,m}^{\az},\theta_{\A,m}^{\el}]\Transpose$} and \smash{$\bm{\theta}_{\D,m} = [\theta_{\D,m}^{\az},\theta_{\D,m}^{\el}]\Transpose$} show the \ac{AoA} and \ac{AoD} in azimuth and elevation (as defined in Fig.
\ref{fig: Local coordinate system and definition of angles} and Section \ref{sec: definitions}), respectively. The array response vectors are given by 
\begin{subequations}
\label{eq:responses}
\begin{align}
    [\mathbf{a}_{\UE}(\bm{\theta}_{\A,m})]_{n} &=\exp \Big( j\frac{2\pi}{\lambda}[\bm{\Delta}_\UE]_{1:3,n}\Transpose \mathbf{d}(\bm{\theta}_{\A,m})\Big), \label{eq:UE array response}\\
    [\mathbf{a}_{\BS}(\bm{\theta}_{\D,m})]_{n} &=\exp\Big( j\frac{2\pi}{\lambda}[\bm{\Delta}_\BS]_{1:3,n}\Transpose \mathbf{d}(\bm{\theta}_{\D,m})\Big), \label{eq:BS array response}
\end{align}
\end{subequations}
in which $\lambda$ is the wavelength at the carrier frequency $f_c$; and $\bm{\Delta}_\UE \triangleq  [\mathbf{x}_{\UE,1},\ldots,\mathbf{x}_{\UE,N_\UE}] \in \mathbb{R}^{3\times N_\UE}$ and $\bm{\Delta}_\BS \triangleq [\mathbf{x}_{\BS,1},\ldots,\mathbf{x}_{\BS,N_\BS}] \in \mathbb{R}^{3\times N_\BS}$ contain coordinates of antenna elements in reference orientation, with respect to the \ac{UE} and the \ac{BS} coordinate frames (i.e., 3D displacements from the phase center), respectively. In \eqref{eq:responses}, we have introduced 
\begin{align} \label{eq: unit-norm direction}
    \mathbf{d}(\bm{\phi}) = [\sin\phi^{\el} \cos\phi^{\az},\sin\phi^{\el} \sin\phi^{\az},\cos\phi^{\el}]\Transpose,
\end{align}
which describes the unit-norm direction of arrival (for $\bm{\phi}=\bm{\theta}_{\A,m}$) and unit-norm direction of departure (for $\bm{\phi}=\bm{\theta}_{\D,m}$) for path $m$, in \ac{UE} and \ac{BS} coordinate frames, respectively.
\vspace{-6mm}
\subsection{Relation Between Geometric Model and Channel Model} \label{sec: definitions}
\allowdisplaybreaks
Corresponding to the \acp{AoA}, the unit-norm arrival directions $\mathbf{d}_{\A,m}=\mathbf{d}(\bm{\theta}_{\A,m})$ are given, in the 
\ac{UE} coordinate frame, by
\begin{align}
    \mathbf{d}_{\A,m} = \begin{cases}
        \R_\UE\Transpose(\p_\BS - \p_\UE)/\Vert\p_\BS - \p_\UE\Vert & m = 0 \\
        \R_\UE\Transpose(\p_m - \p_\UE)/\Vert\p_m - \p_\UE\Vert & m\neq0,
    \end{cases}
\end{align}
which define the \acp{AoA} as $\theta_{\A,m}^{\az} = \atan 2([\mathbf{d}_{\A,m}]_2,[\mathbf{d}_{\A,m}]_1)$ and \smash{$\theta_{\A,m}^{\el} = \acos ([\mathbf{d}_{\A,m}]_3)$}, where $\acos$ is the inverse cosine and $\atan2$ is the four-quadrant inverse tangent. Similarly, using the unit-norm departure directions $\mathbf{d}_{\D,m}=\mathbf{d}(\bm{\theta}_{\D,m})$ given by
\begin{align}
    \mathbf{d}_{\D,m} = \begin{cases}
        \R_\BS\Transpose(\p_\UE - \p_\BS)/\Vert\p_\UE - \p_\BS\Vert & m = 0\\
        \R_\BS\Transpose(\p_m - \p_\BS)/\Vert\p_m - \p_\BS\Vert & m\neq0,
    \end{cases}
\end{align}
in BS coordinate frame, the \acp{AoD} are determined as $\theta_{\D,m}^{\az} = \atan 2([\mathbf{d}_{\D,m}]_2,[\mathbf{d}_{\D,m}]_1)$ and \smash{$\theta_{\D,m}^{\el} = \acos ([\mathbf{d}_{\D,m}]_3)$}. Fig.~\ref{fig: Local coordinate system and definition of angles} shows how AoDs and AoAs and their corresponding directions are defined. Note that the  arrows corresponding to the arrival directions point towards the \ac{BS}/\acp{IP}, and not the \ac{UE}.

Finally, denoting the propagation speed by $c$ and the \emph{unknown} clock bias between the \ac{UE} and the \ac{BS} by $b$, one can express \acp{ToA} as
\begin{align}
    \tau_m=\begin{cases}
        \hfil \Vert\p_\UE-\p_\BS\Vert/c+b & m=0\\
        (\Vert\p_m-\p_\BS\Vert+\Vert\p_\UE-\p_m\Vert)/c+b & m \neq 0.
    \end{cases}
\end{align}
\ColorBlue{The aggregated vectors $\bm{\theta}_\A = [\bm{\theta}_{\A,0}\Transpose,\bm{\theta}_{\A,1}\Transpose,\ldots,\bm{\theta}_{\A,M}\Transpose]\Transpose$,  $\bm{\theta}_\D = [\bm{\theta}_{\D,0}\Transpose,\bm{\theta}_{\D,1}\Transpose,\ldots,\bm{\theta}_{\D,M}\Transpose]\Transpose$,  $\bm{\tau}=[\tau_0,\ldots,\tau_M]\Transpose$, and $\bm{\eta}=[\bm{\theta}_\A\Transpose,\bm{\theta}_\D\Transpose,\bm{\tau}\Transpose]\Transpose$ are defined to be used later.}

\begin{figure}[t!]
    \centering
    \begin{subfigure}[b]{0.45\linewidth}
        \centering
        \begin{tikzpicture}
            \node(image)[anchor=south west] {\includegraphics[height=3cm]{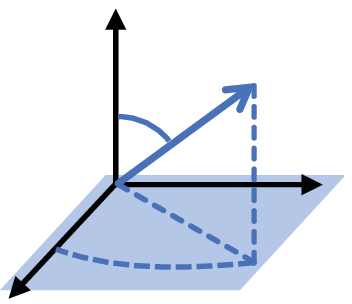}};
            \gettikzxy{(image.north east)}{\ix}{\iy};
            % \draw[help lines] (0,0) grid (\ix,\iy);
            \node at (0.3/4*\ix,0.6/3.2*\iy){$x$};
            \node at (3.9/4*\ix,1.3/3.2*\iy){$y$};
            \node at (1.7/4*\ix,3/3.2*\iy){$z$};
            \node at (1.5/4*\ix,0.8/3.2*\iy){$\theta^{\az}$};
            \node at (1.9/4*\ix,2.3/3.2*\iy){$\theta^{\el}$};
            \node at (3/4*\ix,2.5/3.2*\iy){$\mathbf{d}(\theta)$};
        \end{tikzpicture}
        \vspace{-0.4cm}
        \caption{}
        \label{fig: global coordinate system and reference orientation}
    \end{subfigure}
    \begin{subfigure}[b]{0.45\linewidth}
    \centering
        \begin{tikzpicture}
            \node(image)[anchor=south west] {\includegraphics[height=3cm]{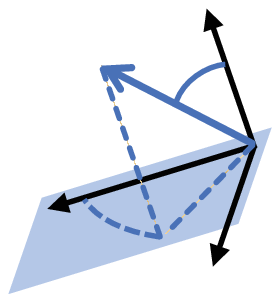}};
            \gettikzxy{(image.north east)}{\ix}{\iy};
            % \draw[help lines] (0,0) grid (\ix,\iy);
            \node at (0.4/3.2*\ix,1.1/3.2*\iy){$x'$};
            \node at (2.7/3.2*\ix,0.5/3.2*\iy){$y'$};
            \node at (2.7/3.2*\ix,2.9/3.2*\iy){$z'$};
            \node at (0.8/3.2*\ix,2.5/3.2*\iy){$\mathbf{d}(\theta)$};
            \node at (2/3.2*\ix,2.6/3.2*\iy){$\theta^{\el}$};
            \node at (1/3.2*\ix,0.7/3.2*\iy){$\theta^{\az}$};
        \end{tikzpicture}
        \vspace{-0.4cm}
        \caption{}
        \label{fig: Local coordinate system and definition of angles}
    \end{subfigure}
\caption{(a) The global coordinate system ($x\!-\!y\!-\!z$) and the reference orientation $\mathbf{R}=\mathbf{I}$, together with the definition of angles and their corresponding direction at orientation $\mathbf{R}=\mathbf{I}$. (b) The local coordinate system ($x'\!-\!y'\!-\!z'$) at a given orientation, together with the definition of angles and their corresponding direction at that orientation. Note that
$\theta^{\az} \in [0,2\pi]$ and $\theta^{\el} \in [0,\pi]$.}
\vspace{-5mm}
\end{figure}
\section{6D Localization Methodology} \label{sec:Methodology}
% \subsection{Two-stage Localization}
\vspace{-10mm}
\ColorBlue{\subsection{Two-stage Localization and Problem Decomposition}}
% We adopt a two-stage localization process, which first involves Channel parameters estimation to determine marginal posterior densities, in the form of estimates of the $M+1$ quintuples $[\bm{\theta}_{\A,m}\Transpose,\bm{\theta}_{\D,m}\Transpose,\tau_m]\Transpose$, along with associated uncertainties. Then, the user location and orientation is estimated, based on the Channel parameters estimation outputs. 
\ColorBlue{We consider a two-stage localization scheme and decompose the problem into a channel parameters estimation routine, followed by a localization routine. The channel estimation routine determines the marginal posterior densities of the channel parameters (in the form of estimates and the associated uncertainties), based on the observations \eqref{eq: signal model}. The localization routine 
uses the output of the channel parameters estimator to determine the 6D state of the \ac{UE}. We consider the contribution of this paper to provide a solution to the second sub-problem, while considering a generic channel parameter estimator. }
%and position this work \emph{not} in the channel parameter estimation domain. However, we still need to make assumptions about that as in the following.}

\subsubsection{Channel Parameters Estimation}
% There exists a variety of channel parameter estimators, including ESPRIT \cite{roemer2014analytical}, generalized approximate message passing \cite{bellili2019generalized}, orthogonal matching pursuit \cite{PositionOrientationEstimation--A.Shahmansoori_others_G.Seco-Granados_H.Wymeersch}, and RIMAX/SAGE \cite{thoma2004rimax}. To keep this work focused on the localization problem, we will assume an efficient Channel parameters estimation routine is present, operating close to the \ac{CRB} (which we derive in Section \ref{sec:FIM Analysis}). From this, we model the estimates as follows. For the \acp{ToA}, we use an independent Gaussian error model \cite{wymeersch2009cooperative}
\ColorBlue{There exists a variety of channel parameter estimators in the literature, including  ESPRIT \cite{roemer2014analytical}, generalized approximate message passing \cite{huang2018iterative,bellili2019generalized}, orthogonal matching pursuit \cite{PositionOrientationEstimation--A.Shahmansoori_others_G.Seco-Granados_H.Wymeersch}, sparse Bayesian learning \cite{cheng2019accurate}, tensor decomposition \cite{ZhouChannelEstimation}, and RIMAX/SAGE \cite{thoma2004rimax}. We assume an arbitrary estimator is applied to obtain an estimate of  $\bm{\eta}$. } 
\ColorBlue{We denote by $\hat{\bm{\eta}}=[\hat{\bm{\theta}}_\A\Transpose,\hat{\bm{\theta}}_\D\Transpose,\hat{\bm{\tau}}\Transpose]\Transpose$, the vector of estimated channel parameters, also referred to as the vector of \emph{measurements}. For each parameter, there is also an associated uncertainty, which leads to the following likelihood functions. For the \ac{ToA} measurements, we assume a multivariate Gaussian distribution \cite{ATE_directionalStatisticsForAngularPositioning,wymeersch2009cooperative}
\begin{align}
\label{eq:TOAlikelihood}
    % p(\hat{\tau}_m|\tau_m)= \exp(-\vert\hat{\tau}_m-\tau_m \vert^2/(2\sigma^2_{m}))/\sqrt{2 \pi \sigma^2_{m}},
    % \hat{\bm{\tau}} \sim \mathcal{N}(\bm{\tau},\Sigma_{\bm{\tau}})
    p(\hat{\bm{\tau}}|\bm{\tau}) & =
    \frac{1}{\sqrt{(2\pi)^{M+1}|\Sigma_{\bm{\tau}}|}}
    e^{-\frac{1}{2}(\hat{\bm{\tau}}-\bm{\tau})\Transpose \Sigma_{\bm{\tau}}^{-1} (\hat{\bm{\tau}}-\bm{\tau})},
    %{\exp\big(\nicefrac{1}{2}(\hat{\bm{\tau}}-\bm{\tau})\Transpose \Sigma_{\bm{\tau}}^{-1} (\hat{\bm{\tau}}-\bm{\tau})\big)}\nicefrac{}{\sqrt{(2\pi)^{M+1}|\Sigma_{\bm{\tau}}|}}
\end{align}
where $\Sigma_{\bm{\tau}}=\mathrm{diag}(\sigma_0^2,\ldots,\sigma_M^2)$. Note that both $\hat{\tau}_m$ and $\sigma^2_{m}$ are provided by the channel parameter estimator.} For the \acp{AoA} and \acp{AoD}, we follow \cite{badiu2017variational} and use a Von Mises distribution, which can be interpreted as a Gaussian distribution over the 1D manifold of angles \cite[Chapter 3]{mardia2009directional}. Correspondingly,\footnote{While it is shown in \cite{ErrorBoundsfor3DLocalization--Z.Abu-Shaban_others_G.Seco-Granados_H.Wymeersch} that azimuth and elevation angles are correlated in general, we neglect the correlation and assume factorized likelihoods, for the sake of tractability. We will later evaluate the impact of this independence assumption (see Section \ref{subsection: obtaining independent parameters}).} 
\begin{subequations} \label{eq: pdf of angles}
\begin{align}
    & p(\hat{\bm{\theta}}_{\A,m}|\bm{\theta}_{\A,m}) = \label{eq:pdf of m-th AoA} \\
    & \exp \left( \kappa_{\A,m}^{\az} \cos (\hat{\theta}^{\az}_{\A,m}-\theta^{\az}_{\A,m}) \right)/\big(2\pi I_0(\kappa_{\A,m}^{\az})\big) \nonumber \\
    \times & \exp \left( \kappa_{\A,m}^{\el} \cos (\hat{\theta}^{\el}_{\A,m}-\theta^{\el}_{\A,m}) \right)/\big(2\pi I_0(\kappa_{\A,m}^{\el})\big), \nonumber \\
    & p(\hat{\bm{\theta}}_{\D,m}|\bm{\theta}_{\D,m}) = \label{eq:pdf of m-th AoD} \\
    & \exp \left( \kappa_{\D,m}^{\az} \cos (\hat{\theta}^{\az}_{\D,m}-\theta^{\az}_{\D,m}) \right) /\big(2\pi I_0(\kappa_{\D,m}^{\az})\big) \nonumber \\
    \times & \exp \left( \kappa_{\D,m}^{\el} \cos (\hat{\theta}^{\el}_{\D,m}-\theta^{\el}_{\D,m}) \right)/\big(2\pi I_0(\kappa_{\D,m}^{\el})\big), \nonumber
\end{align}
\end{subequations}
where $\kappa^{\az}_{\A,m}$, $\kappa^{\el}_{\A,m}$,  $\kappa^{\az}_{\D,m}$, and $\kappa^{\el}_{\D,m}$, are the non-negative concentration parameters of $m$-th \ac{AoA} and $m$-th \ac{AoD}, in azimuth and elevation, respectively, and $I_0(\cdot)$ is the modified Bessel function of order $0$. Again, both the estimated angles and the corresponding concentration parameters are assumed to be provided by the channel parameter estimator.

\ColorBlue{\subsubsection{Location Parameters Estimation}
\label{section maximum likelihood formulation}
The localization problem is formulated as the \acl{ML} estimation
\begin{align}
\label{eq:MLproblemAbstract1}
    % \hat{\bm{\xi}} = \arg\max_{\bm{\xi}}\ln p(\hat{\bm{\eta}}|\bm{\eta}(\bm{\xi})),~\bm{\xi} \in \mathrm{SO}(3) \times \mathbb{R}^{3(M+1)+1},
    \hat{\R}_\UE,\hat{\bm{\zeta}}=\arg\max_{\R_\UE,\bm{\zeta}} \ln p(\hat{\bm{\eta}}|\bm{\eta}(\R_\UE,\bm{\zeta})),~\R_\UE \in \mathrm{SO}(3),
\end{align}
where $\bm{\zeta} = [\p_\UE\Transpose,\p_1\Transpose,\ldots,\p_M\Transpose,b]\Transpose \in \mathbb{R}^{3(M+1)+1}$  
and the likelihood $p(\hat{\bm{\eta}}|\bm{\eta}(\R_\UE,\bm{\zeta}))$ is expressed of the underlying geometry. The likelihood is of the form
\begin{align}
    p\big(\hat{\bm{\eta}}|\bm{\eta}(\R_\UE,\bm{\zeta})\big) = \ &p\big(\hat{\bm{\theta}}_\A|\bm{\theta}_\A (\R_\UE,\bm{\zeta})\big) \\ \times ~ &p\big(\hat{\bm{\theta}}_\D|\bm{\theta}_\D(\bm{\zeta})\big) \times p\big(\hat{\bm{\tau}}|\bm{\tau}(\bm{\zeta})\big),\notag 
\end{align}
in which, by overloading the notation for cosines, we have
\begin{subequations}
\begin{align}
    p\big(\hat{\bm{\theta}}_\A|\bm{\theta}_\A(\R_\UE,\bm{\zeta})\big) &\propto \exp \Big( \bm{\kappa}_{\A}\Transpose \cos\big(\hat{\bm{\theta}}_{\A}-\bm{\theta}_{\A}(\R_\UE,\bm{\zeta})\big) \Big),\\
    p\big(\hat{\bm{\theta}}_\D|\bm{\theta}_\D(\bm{\zeta})\big) &\propto \exp \Big( \bm{\kappa}_{\D}\Transpose \cos\big(\hat{\bm{\theta}}_{\D}-\bm{\theta}_{\D}(\bm{\zeta})\big) \Big),
\end{align}
\end{subequations}
where $\bm{\kappa}_{\A}$ and $\bm{\kappa}_{\D}$ are $2(M+1)\times1$ vectors aggregating the concentration parameters of Von Mises distributions corresponding to the angle estimates in $\hat{\bm{\theta}}_\A$ and $\hat{\bm{\theta}}_\D$, respectively. 
%\smash{$\bm{\kappa}_{\A} = [\kappa^{\az}_{\A,0},\kappa^{\el}_{\A,0},\kappa^{\az}_{\A,1},\kappa^{\el}_{\A,1},\ldots,\kappa^{\az}_{\A,M},\kappa^{\el}_{\A,M}]\Transpose$} and \smash{$\bm{\kappa}_{\D} = [\kappa^{\az}_{\D,0},\kappa^{\el}_{\D,0},\kappa^{\az}_{\D,1},\kappa^{\el}_{\D,1},\ldots,\kappa^{\az}_{\D,M},\kappa^{\el}_{\D,M}]\Transpose$} are the aggregated concentration parameters of Von Mises distributions. 
Considering the negative log-likelihood $\mathcal{L}(\R_\UE,\bm{\zeta}) =-\ln p(\hat{\bm{\eta}}|\bm{\eta}(\R_\UE,\bm{\zeta}))$, the \ac{ML} estimation $\eqref{eq:MLproblemAbstract1}$ can be equivalently written as the constrained minimization
\begin{subequations}
\label{eq:MLproblemAbstract2}
\begin{align}
    \min_{\R_\UE,\bm{\zeta}}~&\mathcal{L}(\R_\UE,\bm{\zeta}),\\
    \mathrm{s.t.}~~&\R_\UE\Transpose\R_\UE=\mathbf{I},~\mathrm{det}(\R_\UE)=+1.\label{unitaryConstraint}
\end{align}
\end{subequations}
where 
\begin{align}
    & \mathcal{L}(\R_\UE,\bm{\zeta})  = \nicefrac{1}{2}(\hat{\bm{\tau}}-\bm{\tau}(\bm{\zeta}))\Transpose \Sigma_{\bm{\tau}}^{-1}(\hat{\bm{\tau}}-\bm{\tau}(\bm{\zeta})) \notag\\
    & - \bm{\kappa}_{\A}\Transpose \cos(\hat{\bm{\theta}}_{\A}-\bm{\theta}_{\A}(\R_\UE,\bm{\zeta})) - \bm{\kappa}_{\D}\Transpose \cos(\hat{\bm{\theta}}_{\D}-\bm{\theta}_{\D}(\bm{\zeta})).
\end{align}
However, solving the optimization problem \eqref{eq:MLproblemAbstract2} using the classical optimization tools is difficult \cite{SPOforVLC--Shen_others_Steendam}, due to the unitary constraint \eqref{unitaryConstraint} on the rotation matrix $\R_\UE$. 
% because (i) it involves optimization over the $\mathrm{SO}(3)$ manifold; (ii) the objective function $\mathcal{L}(.)$ is highly nonlinear and non-convex; (iii) the parameter vector is high-dimensional. 
In addition, the objective function $\mathcal{L}(\R_\UE,\bm{\zeta})$ is highly nonlinear and non-convex, and  iterative algorithms for solving \eqref{eq:MLproblemAbstract2} might reach local optima, if initialized by the points far from the global solution. To address these challenges, we propose a two-step approach: first, we determine an initial ad-hoc estimate of unknowns from geometric arguments, being then refined by a gradient descent of the objective function. The iterative algorithm is explained in the following, while the process to determine an initial point is deferred until Section
\ref{sec:Ad-hoc Initial Estimate}.}

\subsection{An Iterative Algorithm for \ac{ML} Estimation}
As \eqref{eq:MLproblemAbstract2} involves optimization over non-Euclidean manifolds, a suitable optimization tool must be applied. We first present a method for such optimization before applying that to our 6D localization problem. 
\ColorBlue{\subsubsection{Overview of a Method for Optimization over non-Euclidean Manifolds}
Consider the optimization problem 
\begin{align}
     \hat{\mathsf{X}} = \arg \min_{\mathsf{X} \in \mathcal{M}} f(\mathsf{X}),
\end{align}
where $f\!:\!\mathcal{M}\!\to\!\mathbb{R}$ is a smooth function over the (possibly) non-Euclidean manifold $\mathcal{M}$, and $\mathsf{X}$ denotes the (possibly) non-Euclidean parameter}. We exploit the Riemannian gradient descent algorithm, which is a first-order technique analogous to the standard gradient descent algorithm, where the Riemannian gradient is obtained by projection of the classical gradients to the tangent spaces. Starting from the initial point $\hat{\mathsf{X}}^{(0)}$, the algorithm iterates
\begin{align}
    \hat{\mathsf{X}}^{(k+1)}=\mathcal{R}_{ \hat{\mathsf{X}}^{(k)}}\left( -\varepsilon_k \mathcal{P}_{\hat{\mathsf{X}}^{(k)}} \left[ {\partial f(\mathsf{X})}/{\partial \mathsf{X}}\right] _{\mathsf{X}=\hat{\mathsf{X}}^{(k)}}  \right), 
\end{align}
where $\mathcal{P}_{\mathsf{X}}(\cdot)$ is an orthogonal projection onto the tangent space at $\mathsf{X}$, $\mathcal{R}_{\mathsf{X}}(\cdot)$ is a retraction from the tangent space onto $\mathcal{M}$, and $\varepsilon_k>0$ is a suitable step size. Intuitively, the gradient is calculated, and projected to the tangent space (to follow the space of $\mathcal{M}$ as closely as possible), the value at step $k$ is updated, and then the updated value is normalized back into the $\mathcal{M}$.
Relevant for us is $\mathcal{M}=\mathrm{SO}(3)$. Example projection and retractions operations for the Riemannian gradient descent algorithm are given by \cite[Eqs.~(7.32) and (7.22)]{boumal2020introduction}:
\begin{subequations} \label{eq: projection and retraction operators}
\begin{align} 
    \mathcal{P}_{\mathbf{X}}(\mathbf{U})& =\mathbf{X}(\mathbf{X}\Transpose \mathbf{U} - \mathbf{U}\Transpose \mathbf{X})/2,\\
    \mathcal{R}_{\mathbf{X}}(\mathbf{U}) & = (\mathbf{X}+\mathbf{U})(\mathbf{I}_3 + \mathbf{U}\Transpose \mathbf{U})^{-1/2}.
\end{align}
\end{subequations}
On the other hand, for $\mathcal{M}=\mathbb{R}^n$, $\mathcal{P}_{\mathbf{X}}(\mathbf{U})=\mathbf{U}$, $\mathcal{R}_{\mathbf{X}}(\mathbf{U})=\mathbf{X} + \mathbf{U}$, leading to classical gradient descent.

\subsubsection{Solving the ML Estimation Problem} 
\ColorBlue{The optimization variable in our algorithm 
Inspired by the coordinate descent algorithm \cite[Sec. 9.3]{wright1999numerical}, we decompose the unknowns, and apply different optimization algorithms for the estimation of $\R_\UE$ and the rest of unknowns $\bm{\zeta}$,  which belong to the Euclidean space.} %, i.e., $\bm{\zeta} = [\p_\UE\Transpose,\p_1\Transpose,\ldots,\p_M\Transpose,b]\Transpose \in \mathbb{R}^{3(M+1)+1}$.}
We then consider the Riemannian gradient descent algorithm to optimize $\R_\UE$ on the $\mathrm{SO}(3)$ manifold, as
\begin{align}
    \hat{\R}_\UE^{(k+1)} = \mathcal{R}_{\hat{\R}_\UE^{(k)}} \Big( -\varepsilon_k \mathcal{P}_{\hat{\R}_\UE^{(k)}} \big( \left. {\partial \mathcal{L}}/{\partial \R_\UE} \big)\right|_{\hat{\R}_\UE^{(k)},\hat{\bm{\zeta}}^{(k)}} \Big),
\end{align}
with the projection and retraction operators as defined in \eqref{eq: projection and retraction operators}, and $\varepsilon_k$ being the step size obtained using a backtracking line-search \cite{boumal2020introduction}. 
We also use the trust region algorithm \cite[Sec. 11.2]{wright1999numerical}
%standard gradient descent algorithm 
to optimize $\bm{\zeta}$ in the Euclidean space, considering the updated rotation matrix. Then we iterate this algorithm until a stopping criterion is met.
%, according to  $  \hat{\bm{\zeta}}^{(k+1)} = \hat{\bm{\zeta}}^{(k)} - \mu_k {\partial \mathcal{L}}/{\partial \bm{\zeta}}$ where the last terms is evaluated in  ${\hat{\R}_\UE^{(k+1)},\hat{\bm{\zeta}}^{(k)}}$, where we dropped the arguments in $\mathcal{L}(\cdot)$ for simplicity of notation.
% a backtracking line-search \cite{boumal2020introduction}. 
The partial derivatives ${\partial  \mathcal{L}}/{\partial \R_\UE}$ and ${\partial \mathcal{L}}/{\partial \bm{\zeta}}$ are obtained using the chain rule by including the partial derivatives of channel parameters w.r.t. localization parameters given in Appendix \ref{appendix: Partial Derivatives of AoAs, AoDs, and ToAs}. The above method requires initial estimates $\hat{\R}_\UE^{(0)}$ and $\hat{\bm{\zeta}}^{(0)}$, which will be provided by the ad-hoc initial estimator.

\section{Ad-hoc Initial Estimate}
\label{sec:Ad-hoc Initial Estimate}
As stated earlier, the \ac{ML} problem is a nonlinear non-convex optimization, and the elaborated iterative algorithm might reach local optima if the initial values are not given properly. In the following, we propose a simple \ColorBlue{(with linear complexity in the number of \acp{IP} and involving only a 1D line search)} sequential scheme for obtaining initial estimates of localization unknowns based on the estimates of \acp{AoA}, \acp{AoD}, and \acp{ToA}: First, we estimate $\R_\UE$ only from the estimated \acp{AoA} and \acp{AoD}, without any knowledge of $\p_\UE$, $\p_1,\ldots,\p_M$, and $b$. Second, given the estimated $\R_\UE$, we estimate $\p_\UE$, $\p_1,\ldots,\p_M$ from \ac{TDoA} measurements. Finally, we estimate the clock bias $b$ using the estimated positions and \acp{ToA}. The algorithm, therefore, takes the measurements $\hat{\bm{\theta}}_\A$, $\hat{\bm{\theta}}_\D$, and $\hat{\bm{\tau}}$ as input, but we omit the ``hat'' (~$\widehat{.}$~) of the measurements, only in this section, for the notational convenience. Note that the values are not confused with the true \emph{unknown} channel parameters.

\subsection{Step 1: Estimation of UE Rotation Matrix} \label{section:Estimation of the Rotation Matrix}
Our approach to estimate the \ac{UE} rotation matrix from the \acp{AoA} $\bm{\theta}_{\A,m}$ and \acp{AoD} $\bm{\theta}_{\D,m}$ (and corresponding unit vectors $\mathbf{d}({\bm{\theta}}_{\A,m})$ and $\mathbf{d}({\bm{\theta}}_{\D,m})$)
is based on an axis-angle representation of its orientation.
Consider the \ac{LoS} arrival and departure directions $\mathbf{d}_{\A,0}$ and $\mathbf{d}_{\D,0}$, which are along \ac{LoS} path, while in opposite directions in the global coordinate frame. Hence,
\begin{align} \label{eq:AoD-AoA LoS condition}
    \R_\BS \mathbf{d}_{\D,0} = - \R_\UE \mathbf{d}_{\A,0}.
\end{align}
This equation has infinitely many solutions for $\R_\UE$, satisfying both \eqref{eq:AoD-AoA LoS condition} and the orthogonality constraint, among which, one is the true $\R_\UE$. 

\subsubsection{Characterizing the Solutions to  \eqref{eq:AoD-AoA LoS condition}} 
We first find a solution for $\R_\UE$ in \eqref{eq:AoD-AoA LoS condition}, denoted by $\widetilde{\R} \in \text{SO}(3)$. Multiplying $\widetilde{\R}$ with a rotation $\psi \in [0,2\pi)$ around $\mathbf{d}_{\A,0}$ yields all rotation matrices that satisfy \eqref{eq:AoD-AoA LoS condition}.
\begin{lemma} \label{lem:lemma1}
One solution for \eqref{eq:AoD-AoA LoS condition} is given by 
\begin{align} \label{eq:solution R_tilde}
    \widetilde{\R} = \mathbf{I} + [\mathbf{d}]_{\times} + \frac{1}{1-\mathbf{d}_{\A,0}\Transpose \R_\BS \mathbf{d}_{\D,0}} [\mathbf{d}]_{\times}^2,
\end{align}
where $\mathbf{d}=[d_1,d_2,d_3]\Transpose \triangleq -\mathbf{d}_{\A,0} \times \R_\BS \mathbf{d}_{\D,0}$, and 
\begin{align}\label{eq:skew-symmetric matrix}
    [\mathbf{d}]_{\times} \triangleq
    \begin{bmatrix}
        \phantom{-}0 & -d_3 & \phantom{-}d_2 \\
        \phantom{-}d_3 & \phantom{-}0 & -d_1 \\
        -d_2 & \phantom{-}d_1 & \phantom{-}0
    \end{bmatrix}.
\end{align}
\end{lemma}
\begin{proof}
Let us define $\tilde{\mathbf{d}}=-\R_\BS \mathbf{d}_{\D,0}$.
The result follows from Rodrigues' formula, by rotating $\mathbf{d}_{\A,0}$ to $\tilde{\mathbf{d}}$ with rotation axis $\mathbf{d}_{\A,0} \times \tilde{\mathbf{d}}$ \cite[Section 9.6.2]{vince2011rotation}.
\end{proof}
\begin{lemma}
The transformation matrix describing rotations by the angle $\psi \in [0,2\pi)$ around the arbitrary unit-norm vector $\mathbf{u}$ is given by  $\mathbf{Q}_{\mathbf{u}}(\psi) = 
     [\mathbf{u}]_{\times} \sin \psi + (\mathbf{I}-\mathbf{u}\mathbf{u}\Transpose) \cos \psi + \mathbf{u}\mathbf{u}\Transpose$.
\end{lemma}
\begin{proof}
The result follows from Rodrigues' formula \cite[Section 9.6.2]{vince2011rotation}.
\end{proof}
The rotation matrices $\R_\UE$ satisfying \eqref{eq:AoD-AoA LoS condition} 
are thus characterized as
\begin{align} \label{eq:characterization of rotation matrix}
    \R(\psi) =  \widetilde{\R}\mathbf{Q}_{\mathbf{d}_{\A,0}}(\psi),~\forall \psi \in [0,2\pi).
\end{align}
It is easily verified that \eqref{eq:AoD-AoA LoS condition} holds for $\forall \psi \in [0,2\pi)$, since $-\widetilde{\R}\mathbf{Q}_{\mathbf{d}_{\A,0}}(\psi)\mathbf{d}_{\A,0}  = -\widetilde{\R}\mathbf{d}_{\A,0} = \R_\BS \mathbf{d}_{\D,0}$, 
where the first transition is due to a rotation around an axis leaving that axis invariant, and the second transition due to Lemma \ref{lem:lemma1}. What remains is now to determine $\psi$ based on the \ac{NLoS} paths. 

\begin{figure}
%[ht!]
%     \begin{subfigure}[b]{0.5\textwidth}
%     \centering
%     \begin{tikzpicture}
%     \node(image)[anchor=south west] {\includegraphics[height=4cm]{figures/fig_3a.pdf}};
%     \gettikzxy{(image.north east)}{\ix}{\iy};
%     % \draw[help lines] (0,0) grid (\ix,\iy);
%     \node at (1.8/7*\ix,0.2/4*\iy){$y_\UE$};
%     \node at (1.8/7*\ix,3.7/4*\iy){$z_\UE$};
%     \node at (0.2/7*\ix,1.4/4*\iy){$x_\UE$};
%     \node at (1.5/7*\ix,2.2/4*\iy){$\mathbf{d}_{\A,0}$};
%     \node at (0.4/7*\ix,3.4/4*\iy){LoS};
%     %************************************
%     \node at (5.4/7*\ix,3.8/4*\iy){$y_\UE$};
%     \node at (4/7*\ix,1.7/4*\iy){$z_\UE$};
%     \node at (7/7*\ix,3.4/4*\iy){$x_\UE$};
%     \node at (4.5/7*\ix,2.2/4*\iy){$\mathbf{d}_{\A,0}$};
%     \node at (3.4/7*\ix,3.4/4*\iy){LoS};
%     \end{tikzpicture}
%     \caption{}
%     \label{fig: rotation around LoS}
% \end{subfigure}
    % \hfill 
    %\begin{subfigure}[b]{0.5\textwidth}
    \centering
    \begin{tikzpicture}
    \node(image)[anchor=south west] {\includegraphics[height=4cm]{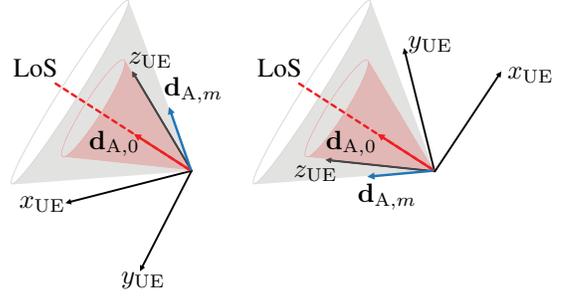}};
    \gettikzxy{(image.north east)}{\ix}{\iy};
    % \draw[help lines] (0,0) grid (\ix,\iy);
    \node at (2.1/7.5*\ix,0.2/4.2*\iy){$y_\UE$};
    \node at (2.2/7.5*\ix,3.1/4.2*\iy){$z_\UE$};
    \node at (0.7/7.5*\ix,1.2/4.2*\iy){$x_\UE$};
    \node at (2.8/7.5*\ix,2.7/4.2*\iy){$\mathbf{d}_{\A,m}$};
    \node at (1.7/7.5*\ix,2/4.2*\iy){$\mathbf{d}_{\A,0}$};
    \node at (0.6/7.5*\ix,3/4.2*\iy){LoS};
    %************************************
    \node at (6.1/7.5*\ix,3.3/4.2*\iy){$y_\UE$};
    \node at (4.5/7.5*\ix,1.6/4.2*\iy){$z_\UE$};
    \node at (7.5/7.5*\ix,2.9/4.2*\iy){$x_\UE$};
    \node at (5.5/7.5*\ix,1.3/4.2*\iy){$\mathbf{d}_{\A,m}$};
    \node at (5/7.5*\ix,2/4.2*\iy){$\mathbf{d}_{\A,0}$};
    \node at (4/7.5*\ix,3/4.2*\iy){LoS};
    \end{tikzpicture}
    %\caption{}
%\end{subfigure}
\caption{%\setstretch{1.2}
%(a) The \ac{LoS} arrival direction $\mathbf{d}_{\A,0}$ is left invariant, and the axis $z_\UE$ rotates, in the global coordinate frame, on the lateral surface of a cone (the red cone in this figure) with axis $\mathbf{d}_{\A,0}$ and apex angle $\theta_{\A,0}^{\el}$, when the rotation around axis $\mathbf{d}_{\A,0}$ is applied. The figure illustrates two snapshots of such rotation. (b) 
The example \ac{NLoS} arrival direction $\mathbf{d}_{\A,m}$ rotates, in the global coordinate frame, while maintaining the same angle $\bm{\theta}_{\A,m}$ in \ac{UE} coordinate frame, on the lateral surface of a cone (the gray cone in this figure) with axis $\mathbf{d}_{\A,0}$, and apex angle equal to the angle between $\mathbf{d}_{\A,0}$ and $\mathbf{d}_{\A,m}$, when the rotation around axis $\mathbf{d}_{\A,0}$ is applied. The figure illustrates two snapshots of such rotation.} \vspace{-4mm}
\label{fig:findingAngles}
\end{figure}
\subsubsection{Rotation Estimation Based on \ac{NLoS} \acp{AoA} and \acp{AoD}} We now determine the angle $\psi\in[0,2\pi)$, so that the combined rotation resulting from $\widetilde{\R}$ and $\mathbf{Q}_{\mathbf{d}_{\A,0}}(\psi)$, i.e., $\R(\psi)=\widetilde{\R}\mathbf{Q}_{\mathbf{d}_{\A,0}}(\psi)$, leads to the arrival directions $\mathbf{d}_{\A,m},~m=1,\ldots,M$.
To determine $\psi$, we note the following (see Fig.~\ref{fig:schematic of system model with LoS and NLoS paths}):
\begin{itemize}
    \item The departure directions $\mathbf{d}_{\D,m},~m=1,\ldots,M$, determine the half-lines $\ell_{\D,m}(\p_\BS,\R_\BS,\mathbf{d}_{\D,m})=\{ \p\in \mathbb{R}^3: \p=\p_\BS + t_{\D,m} \R_\BS \mathbf{d}_{\D,m}$, $t_{\D,m} \ge 0\}$, in the global coordinate frame. 
    \item Under UE rotation represented by $\R({\psi})$, the arrival directions $\mathbf{d}_{\A,m},~m=1,\ldots,M$, determine the half-lines $\ell_{\A,m} (\p_\UE,\R(\psi),\mathbf{d}_{\A,m})=\{ \p\in \mathbb{R}^3: \p=\p_\UE + t_{\A,m} \R(\psi)\mathbf{d}_{\A,m}$, $t_{\A,m} \ge 0\}$, for any given $\psi \in [0,2\pi)$, in the global coordinate frame. See Fig.~\ref{fig:findingAngles}.
    \item With the correct $\psi$, the half-lines $\ell_{\D,m}(\p_\BS,\R_\BS,\mathbf{d}_{\D,m})$ and $\ell_{\A,m}(\p_\UE,\R(\psi),\mathbf{d}_{\A,m})$ intersect at the incidence point $\p_m$ (in the absence of noise). 
\end{itemize}
However, (i) neither $\p_\UE$ nor $\p_m$ are known; (ii)
the half-lines might not necessarily intersect due to the noisy measurements.\footnote{They form a pair of so-called \emph{skew lines}, i.e., lines in 3D that do not intersect, while not being parallel.} To tackle the first challenge, we note that the argument that the half-lines intersect under the correct value of $\psi$ is true for any scaling of the global coordinate system. 
We express $\p_\UE = \p_\BS + \rho_0 \R_\BS \mathbf{d}_{\D,0}$, where $\rho_0=\Vert\p_\BS-\p_\UE\Vert$. Any point $\mathbf{p}\in \mathbb{R}^3$ can be expressed as $\mathbf{p} = \p_\BS + \rho \R_\BS \mathbf{d}_{\D}$, with suitable $\rho$ and $\mathbf{d}_{\D}$, where $\mathbf{d}_{\D}$ is a unit-norm vector. Given any $r >0$, we can define a scaled system (with scaling $r/\rho_0$) with the BS location as an invariant point:
\begin{align}
    \p(r) & = \p_\BS + \frac{r}{\rho_0} \rho \R_\BS \mathbf{d}_{\D}, \label{eq:p_UE in scaled coordinate system2}
\end{align}
and in particular
\begin{align}
    \p_\UE(r) &  = \p_\BS + \frac{r}{\rho_0} \rho_0 \R_\BS \mathbf{d}_{\D,0}.  \label{eq:p_UE in scaled coordinate system} %\\
\end{align}
Without loss of generality, we set $r=1$ so that $\p_\UE(r=1)$ is known.  
Hence, %any point $\mathbf{p}'\in\p_\BS + \rho' \R_\BS \mathbf{d}'_{\D}$ in the original coordinate system becomes $\mathbf{p}'\in\p_\BS + \rho' \R_\BS \mathbf{d}'_{\D}$
%Having expressed the \ac{UE} position as
%\begin{align} \label{eq: expression of p_UE in %terms of the distance of BS and UE}
%    \p_\UE = \p_\BS + \rho_0 \R_\BS \mathbf{d}_{\D,0}, 
%\end{align}
%where $\rho_0=\Vert\p_\BS-\p_\UE\Vert$, facilitates definition of a scaled coordinate system as
%\begin{align} \label{eq:p_UE in scaled coordinate system}
 %   \p_\UE^\prime(\rho_0^\prime) = {\p_\UE}_{|\rho_0 \rightarrow \rho_0^\prime} = \p_\BS + \frac{\rho_0^\prime}{\rho_0} \rho_0\R_\BS \mathbf{d}_{\D,0}, 
%\end{align}
%for any $\rho_0^\prime>0$. This scaled coordinate system  characterizes auxiliary half-lines 
%$\ell_{\A,m}(\p_\UE^\prime(r),\R(\psi),\mathbf{d}_{\A,m})$ is parallel to $\ell_{\A,m}(\p_\UE,\R(\psi),\mathbf{d}_{\A,m})$, for $m=1,\ldots,M$, and thus 
in the scaled coordinate system with $r=1$, 
$\ell_{\A,m}(\p_\UE(r),\R(\psi),\mathbf{d}_{\A,m})$ and $\ell_{\D,m}(\p_\UE(r),\R(\psi),\mathbf{d}_{\D,m})$ intersect for the correct value of $\psi$, $\forall m$, in the absence of noise. 
%The choice of $\rho_0^\prime$ is arbitrary so that, without LoSs of generality, we set $\rho_0^\prime=1$, which can be interpreted as scaling the geometric  system model from Section \ref{subsec:geometric model} by $1/\rho_0$, maintaining the same \acp{AoA} and \acp{AoD} (see Fig.~\ref{fig: scaling system geometry}). 
To cope with the second challenge (measurement noise), we use the distance between skew lines in the least squares objective. In particular, the shortest distance $\delta_m(\psi)$ between the half-lines $\ell_{\D,m}(\p_\BS,\R_\BS,\mathbf{d}_{\D,m})$ and $\ell_{\A,m}(\p_\UE^\prime,\R(\psi),\mathbf{d}_{\A,m})$ is obtained from the solution of parametric optimization 
\begin{align}\label{opt:shortest distance}
    \delta_m^2(\psi) = \min_{\mathbf{t}_m\ge \mathbf{0}} D_m^2(\mathbf{t}_m,\R(\psi)),~m>0,
\end{align}
where $\mathbf{t}_m=[t_{\D,m},t_{\A,m}]\Transpose$ and $ D_m^2(\mathbf{t}_m,\R(\psi)) = \Vert (\p_\BS + t_{\D,m} \R_\BS \mathbf{d}_{\D,m}) - (\p_\UE(1) + t_{\A,m} \R(\psi) \mathbf{d}_{\A,m}) \Vert^2$. The optimization \eqref{opt:shortest distance} is a quadratic convex problem, and the solution of that is provided in Appendix \ref{appendix: solving opt for shortest distance}. Combining the minimum distances, we estimate $\psi$ as 
\begin{align} 
\label{eq:estimation of psi}
    \hat{\psi} = \arg \min_{\psi \in [0,2\pi)} \Vert\bm{\delta}(\psi)\Vert,
\end{align}
where $\bm{\delta}(\psi)=[\delta_1(\psi),\ldots,\delta_M(\psi)]\Transpose$, and accordingly, the estimate of $\R_\UE$ is given by $\hat{\R}_\UE = \R(\hat{\psi})$, 
% \begin{align} \label{eq: estimate of UE orientation}
%     \hat{\R}_\UE = \R(\hat{\psi}),
% \end{align}
with $\R(\psi)$ characterized as in \eqref{eq:characterization of rotation matrix}. %Algorithm \ref{algorithm: Initial Estimate of Rotation Matrix} summarizes the proposed method.

\subsection{\ColorBlue{Step 2: Estimation of Positions}}
We obtain the auxiliary points $\p_m(r=1)$ in the scaled coordinate system induced by \eqref{eq:p_UE in scaled coordinate system2}, as the nearest points to the half-lines $\ell_{\D,m}(\p_\BS,\R_\BS,\mathbf{d}_{\D,m})$ and $\ell_{\A,m}(\p_\UE(1),\hat{\R}_\UE,\mathbf{d}_{\A,m})$ in a least-squares sense (considering the lines from UE and BS as full lines, rather than half-lines, for computational complexity considerations), using 
\begin{subequations} \label{eq: estimate of nearest points}
\begin{align}
    \p_m(r=1) &= \mathbf{A}_m^{-1} \mathbf{b}_m, \\
    \mathbf{A}_m &= \mathbf{P}_{\perp}(\R_\BS\mathbf{d}_{\D,m})+\mathbf{P}_{\perp}(\hat{\R}_\UE\mathbf{d}_{\A,m}),\\
    \mathbf{b}_m &= \mathbf{P}_{\perp}(\R_\BS\mathbf{d}_{\D,m})\p_\BS %\\ &
    +\mathbf{P}_{\perp}(\hat{\R}_\UE\mathbf{d}_{\A,m})\p_\UE(1) %\nonumber.
\end{align}
\end{subequations}
where $\mathbf{P}_{\perp}(\mathbf{d}) \triangleq \mathbf{I} - \mathbf{d}\mathbf{d}\Transpose$
is the projector onto the subspace orthogonal to the one spanned by a vector $\mathbf{d}$ (see Appendix \ref{proof: nearest point to skew lines}). Then, we define the auxiliary departure direction $\mathbf{d}_{\D,m}^\prime$ towards $\p_m(1)$ (in \ac{BS} coordinate frame), and the corresponding distance $\rho_m^\prime$ as (see Fig.~\ref{fig: visulization of direction from BS to the estimates of incidence points positions})
\begin{subequations} \label{eq: direction and distance of nearest points}
\begin{align}
    \mathbf{d}_{\D,m}^\prime &= \R_\BS\Transpose \frac{\p_m(1)-\p_\BS}{\Vert\p_m(1)-\p_\BS\Vert},\\
    \rho_m^\prime &= \Vert\p_m(1)-\p_\BS\Vert,
\end{align}
\end{subequations}
allowing us to express $\p_m(1) = \p_\BS + \rho_m^\prime \R_\BS \mathbf{d}_{\D,m}^\prime$.
% \begin{align} \label{eq:p_m in scaled coordinate system}
%     \p_m^\prime(\rho_0^\prime) = \p_\BS + \rho_m^\prime \R_\BS \mathbf{d}_{\D,m}^\prime,
% \end{align}
%(see Fig.~\ref{fig: visulization of direction from BS to the estimates of incidence points positions}). 
What remains is to find $\rho_0$, since, with knowledge of $\rho_0$, we can scale the UE and IP positions to their correct place, keeping the BS fixed. To see this, consider triangles with the vertices $\{\p_\BS,\p_\UE(1),\p_m(1)\}$, $m>1$, and scale them by $\rho_0$, yielding
\begin{align}
    \p_m(\rho_0) & = \p_\BS + {\rho_0}\rho_m^\prime \R_\BS \mathbf{d}_{\D,m}^\prime \\
    \p_\UE(\rho_0) & = \p_\BS + \frac{\rho_0}{\rho_0} \rho_0 \R_\BS \mathbf{d}_{\D,0}.
\end{align}
%where the approximation is due to the measurement noise. 
%
% Now, let us consider triangles with the vertices $\{\p_\BS,\p_\UE(1),\p_m(1)\}$, $m=1,\ldots,M$, and scale them by $\rho_0$, effectively scaling the entire coordinate system again to its original size, immediately reco. For the second scale, we obtain
% \begin{align} \label{eq: expression of p_m in terms of the distance of BS and UE}
%     \p_m^{\prime\prime}(\rho_0) = {\p_m^\prime(\rho_0^\prime)}_{|\rho_0^\prime \rightarrow \rho_0} = \p_\BS + \frac{\rho_0}{\rho_0^\prime}\rho_m^\prime \R_\BS \mathbf{d}_{\D,m}^\prime.
% \end{align}
% This double-scaling gives the original coordinate system, where clearly
% \begin{align} 
%     \p_\UE^{\prime\prime}(\rho_0) = {\p_\UE^\prime(\rho_0^\prime)}_{|\rho_0^\prime \rightarrow \rho_0} = \p_\UE.
% \end{align}
% Hence $\p_m^{\prime\prime}(\rho_0)$ can provide estimates of incidence points $\p_m$. One should note that $\p_\UE = \p_\UE^\prime(\rho_0)$ and $\p_m=\p_m^{\prime\prime}(\rho_0)$ are true in noiseless case, and give estimates of positions, as the measurements are noisy in general.
% \begin{remark}
%     The auxiliary departure direction $\mathbf{d}_{\D,m}^\prime$ is not a real signal departure direction, and is introduced for analogy between the expressions for $\p_\UE^\prime(\rho_0^\prime)$ and $\p_m^\prime(\rho_0^\prime)$.
% \end{remark}
\begin{figure}[t!]
    \centering
    \begin{tikzpicture}
    \node(image)[anchor=south west] {\includegraphics[height=3.2cm]{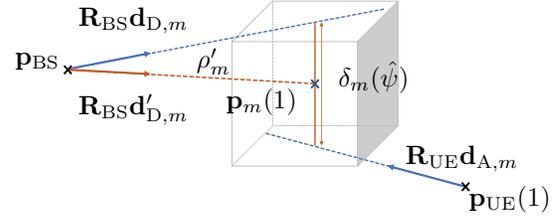}};
    \gettikzxy{(image.north east)}{\ix}{\iy};
    % \draw[help lines] (0,0) grid (\ix,\iy);
    \node at (0.3/6.5*\ix,2/3.2*\iy){$\p_\BS$};
    \node at (6.2/6.5*\ix,0.3/3.2*\iy){$\p_\UE(1)$};
    \node at (3.1/6.5*\ix,1.5/3.2*\iy){$\p_m(1)$};
    \node at (1.5/6.5*\ix,2.5/3.2*\iy){$\R_\BS\mathbf{d}_{\D,m}$};
    \node at (1.5/6.5*\ix,1.4/3.2*\iy){$\R_\BS\mathbf{d}^\prime_{\D,m}$};
    \node at (5.6/6.5*\ix,0.8/3.2*\iy){$\R_\UE\mathbf{d}_{\A,m}$};
    \node at (2.5/6.5*\ix,2/3.2*\iy){$\rho^\prime_m$};
    \node at (4.5/6.5*\ix,1.8/3.2*\iy){$\delta_m(\hat{\psi})$};
    \end{tikzpicture}
    \caption{Visualization of auxiliary direction $\mathbf{d}_{\D,m}^\prime$ and auxiliary points  $\p_\UE(1)$ and $\p_m(1)$ in the scaled geometric model.}
    \label{fig: visulization of direction from BS to the estimates of incidence points positions}
\end{figure}
To recover the value of  ${\rho}_0$, we rely on  the \ac{TDoA} measurements. Introducing $\bm{\Delta}=[\Delta_1,\ldots,\Delta_M]\Transpose$ with $\Delta_m=c(\tau_m - \tau_0)$, and $\tilde{\bm{\Delta}}(\rho_0) =[\tilde{\Delta}_1(\rho_0),\ldots,\tilde{\Delta}_M(\rho_0)]\Transpose$, where%, considering $\rho_0^\prime=1$,
\begin{align}
    \tilde{\Delta}_m(\rho_0)  & =\Vert\p_m(\rho_0)-\p_\BS\Vert+\Vert \p_\UE(\rho_0)-\p_m(\rho_0)\Vert-\rho_0 \notag \\
    & =\rho_0 (\rho_m^\prime + \Vert\mathbf{d}_{\D,0}-\rho_m^\prime\mathbf{d}_{\D,m}^\prime\Vert-1),
\end{align}
allows us to formulate 
\begin{align}\label{eq:least squares for the distance between UE and BS}
    \hat{\rho}_0 = \arg \min_{\rho_0} \Vert \bm{\Delta} - \tilde{\bm{\Delta}}(\rho_0) \Vert^2.
\end{align}
Since $\tilde{\bm{\Delta}}(\rho_0)=\rho_0\bm{\beta}$ where $\beta_m=\rho_m^\prime + \Vert\mathbf{d}_{\D,0}-\rho_m^\prime\mathbf{d}_{\D,m}^\prime\Vert-1$, the optimization \eqref{eq:least squares for the distance between UE and BS} is quadratic in $\rho_0$ and gives the closed-form solution $\hat{\rho}_0 = {\bm{\beta}\Transpose \bm{\Delta}}/{\bm{\beta}\Transpose \bm{\beta}}$. Note that $\beta_m\ge 0$ since $\tilde{\Delta}_m (\rho_0)\ge 0,~\forall\rho_0$, according to the Triangle inequality. Hence,  $\hat{\rho}_0\ge 0$, and the solution is  meaningful. The estimates of positions are then given by $\hat{\p}_\UE = \p_\UE(\hat{\rho}_0)$ and $ \hat{\p}_m = \p_m(\hat{\rho}_0)$, $m=1,\ldots,M$.
% \begin{subequations}\label{eq: estimate of positions}
% \begin{align} 
%     \hat{\p}_\UE &= \p_\UE(\hat{\rho}_0),\\
%     \hat{\p}_m &= \p_m(\hat{\rho}_0),~m=1,\ldots,M.
% \end{align}
% \end{subequations}
%Algorithm \ref{algorithm: Initial Estimate of Locations} summarizes the proposed method.

% \begin{algorithm}[t]
% \caption{Initial Estimate of Positions}
% \label{algorithm: Initial Estimate of Locations}
% \hspace*{\algorithmicindent} \textbf{Input:} $\hat{\R}_\UE,~ \hat{\tau}_m,~m=0,1,\ldots,M$.  
% \begin{algorithmic}[1]
%     \State Obtain $\p_m^\prime(\rho_0^\prime)$ using \eqref{eq: estimate of nearest points}.
%     \State Obtain $\mathbf{d}_{\D,m}^\prime$ and $\rho_m^\prime$ according to \eqref{eq: direction and distance of nearest points}. 
%     \State Find $\hat{\rho}_0$ according to \eqref{eq:estimate of the distance between UE and BS}.
%     \State Estimate $\hat{\p}_\UE$ and $\hat{\p}_m,~m=1,\ldots,M$, using \eqref{eq: estimate of positions}. 
% \end{algorithmic}
% \hspace*{\algorithmicindent} \textbf{Output:} $\hat{\p}_\UE,~ \hat{\p}_m,~m=1,\ldots,M$.
% \end{algorithm}

\subsection{Step 3: Estimation of Clock Bias}
After estimating the positions for \ac{UE} and incidence points, we estimate the clock bias as
\begin{align} \label{opt: estimation of clock bias}
    \hat{b} = \arg \min_b \Vert \bm{\tau}-\hat{\bm{\tau}}(\hat{\rho}_0)+ b\mathbf{1}_{M+1}\Vert^2,
\end{align}
with $\bm{\tau}=[\tau_0,\ldots,\tau_M]\Transpose$, $\hat{\bm{\tau}}(\hat{\rho}_0)=[\hat{\tau}_0(\hat{\rho}_0),\ldots,\hat{\tau}_M(\hat{\rho}_0)]\Transpose$, and
\begin{align}
    \hat{{\tau}}_m(\hat{\rho}_0)= \begin{cases}
        \hfil \Vert \hat{\p}_\UE - \p_\BS\Vert/c = \hat{\rho}_0/c & m=0 \\
        (\Vert\hat{\p}_m-\p_\BS\Vert+\Vert\hat{\p}_\UE-\hat{\p}_m\Vert)/c & m\neq 0,
    \end{cases}\notag
\end{align}
giving the closed-form solution  
$\hat{b}={\mathbf{1}\Transpose(\bm{\tau}-\hat{\bm{\tau}}(\hat{\rho}_0))}/{(M+1)}$.

\section{Fisher Information Analysis}% and Performance Error Bounds}
\label{sec:FIM Analysis}
In this section, we derive the \ac{FIM} of the channel parameters and the localization parameters, and obtain the error bounds for 6D localization, mapping, as well as \ac{UE} synchronization.

\subsection{\ac{FIM} of Channel Parameters}
We define the vector of channel parameters as 
\begin{align}
    \bm{\eta}_{\mathrm{ch}} \triangleq [\underbrace{\bm{\theta}_\A\Transpose,\bm{\theta}_\D\Transpose,\bm{\tau}\Transpose}_{\bm{\eta}\in \mathbb{R}^{5(M+1)}},\mathbf{h}_{\mathrm{R}}\Transpose,\mathbf{h}_{\mathrm{I}}\Transpose]\Transpose \in \mathbb{R}^{7(M+1)}, 
\end{align} 
where $\mathbf{h}_{\mathrm{R}} = [h_{\mathrm{R},0},\ldots,h_{\mathrm{R},M}]\Transpose$, and $\mathbf{h}_{\mathrm{I}} = [h_{\mathrm{I},0},\ldots,h_{\mathrm{I},M}]\Transpose$. The \ac{FIM} of $\bm{\eta}_{\mathrm{ch}}$, considering the signal model \eqref{eq: signal model}, is given by the Slepian-Bangs formula \cite[Section 3.9]{kay1993fundamentals} as
\begin{align} \label{eq: expression for FIm of channel}
    [\mathbf{J}_{\bm{\eta}_{\mathrm{ch}}}]_{i,j}=\frac{2E_\mathrm{s}}{N_0}\sum_{k=1}^K\sum_{n=1}^N \Re \left\{ \frac{\partial\widetilde{y}_{k,n}\Hermitian}{\partial [{\bm{\eta}_{\mathrm{ch}}}]_{i}}\frac{\partial\widetilde{y}_{k,n}}{\partial [{\bm{\eta}_{\mathrm{ch}}}]_{j}} \right\},
\end{align}
where $\widetilde{y}_{k,n}$ is the noise-free part of the observation $y_{k,n}$, and the gradients can be found in \cite[Appendix.~\MakeUppercase{\romannumeral 1}]{ErrorBoundsfor3DLocalization--Z.Abu-Shaban_others_G.Seco-Granados_H.Wymeersch}. Then we obtain the \ac{EFIM} of \acp{AoA}, \acp{AoD}, and \acp{ToA} as in the following:
\begin{align} \label{eq: EFIM of angles and delays}
\mathbf{J}_{\bm{\eta}} = \left[[\mathbf{J}_{\bm{\eta}_{\mathrm{ch}}}^{-1}]_{1:5(M+1),1:5(M+1)}\right]^{-1}.
\end{align}

\subsection{FIM of Localization Parameters}
Expressing $\R_\UE=[\vecR_{\UE,1},\vecR_{\UE,2},\vecR_{\UE,3}]$ with $\vecR_{\UE,1}$, $\vecR_{\UE,2}$, and $\vecR_{\UE,3}$ as columns, we recall the vector of localization unknowns as
\begin{align} \label{eq:vector of localization unknowns}
    \bm{\xi} = [\vecR\Transpose,\p_\UE\Transpose,\p_1\Transpose,\ldots,\p_M\Transpose,b]\Transpose \in \mathbb{R}^{3(M+1)+10}.
\end{align} 
where $\vecR=\mathrm{vec}(\R_\UE)=[\vecR_{\UE,1}\Transpose,\vecR_{\UE,2}\Transpose,\vecR_{\UE,3}\Transpose]\Transpose$.
Then we obtain the \ac{FIM} $\mathbf{J}_{\bm{\xi}}$ by transforming the channel parameters to localization parameters through the Jacobian matrix $\bm{\Upsilon}$ as follows \cite[Eq.~(3.30)]{kay1993fundamentals}:
\begin{align} \label{eq:unconstrained FIM}
    \mathbf{J}_{\bm{\xi}}= \bm{\Upsilon}\Transpose \mathbf{J}_{\bm{\eta}} \bm{\Upsilon},
\end{align}
where $[\bm{\Upsilon}]_{i,j} = \partial [\mathbf{J}_{\bm{\eta}}]_i / \partial [\mathbf{J}_{\bm{\xi}}]_j$, and $\mathbf{J}_{\bm{\eta}}$ is given in \eqref{eq: EFIM of angles and delays}. To obtain the elements of the transformation matrix $\bm{\Upsilon}$, we need the derivatives of channel parameters w.r.t. localization parameters, which are obtained in Appendix \ref{appendix: Partial Derivatives of AoAs, AoDs, and ToAs}. We note that ${\partial [{\bm{\eta}}]_i}/{\partial \vecR} = \vect\left({\partial [{\bm{\eta}}]_i}/{\partial \R_\UE}\right)$. 

To obtain the error bounds, we need to account for the constraint that $\R_\UE \in \text{SO}(3)$. We obtain the \ac{CCRB} \cite{CCRB_Stoica} giving the lower bound on the estimation error covariance, for any unbiased estimator subject to the required constraint on $\R_\UE$. The set of constraints to be satisfied due to orthogonality of the rotation matrix (i.e., $\R_\UE\Transpose \R_\UE=\mathbf{I}_3$) is given by
\begin{align}
    \mathbf{h}(\mathbf{\bm{\xi}}) = [& \Vert\vecR_1\Vert^2  - 1 ,  \vecR_2\Transpose \vecR_1 , \vecR_3\Transpose \vecR_1, \nonumber \\ &
    \Vert\vecR_2\Vert^2 - 1 , \vecR_2\Transpose \vecR_3 , \Vert\vecR_3\Vert^2-1]\Transpose = \mathbf{0}_{6}.
\end{align}
Considering $\mathbf{M} = \mathrm{blkdiag}(\frac{1}{\sqrt{2}}\mathbf{M}_0,\mathbf{I}_{3(M+1)+1})$ with
\begin{align} 
\mathbf{M}_0 = 
\begin{bmatrix}
\begin{array}{rrrr}
        -\vecR_3 & \mathbf{0}_3 & \vecR_2\\
        \mathbf{0}_3 & -\vecR_3 & -\vecR_1\\
        \vecR_1 & \vecR_2 & \mathbf{0}_3\\
\end{array}
\end{bmatrix} \in \mathbb{R}^{9\times3},
\end{align} 
meets $\mathbf{G}(\bm{\xi})\mathbf{M}=\mathbf{0}$ where $[\mathbf{G}(\bm{\xi})]_{i,j} = {\partial [\mathbf{h}(\bm{\xi})}]_i/{\partial [\bm{\xi}]_j}$, and gives \cite{CCRB_Stoica}
\begin{align} \label{eq:CCRBfinal}
    \mathbf{C}\CCRB_{\bm{\xi}} = \mathbf{M} (\mathbf{M}\Transpose \mathbf{J}_{\bm{\xi}} \mathbf{M})^{-1}\mathbf{M}\Transpose,
\end{align}
in which $\mathbf{J}_{\bm{\xi}}$ is given in \eqref{eq:unconstrained FIM}. Then any unbiased estimate $\hat{\bm{\xi}}$ subject to $\hat{\R}_\UE \in \text{SO}(3)$ satisfies $\mathbb{E} \big\{ (\hat{\bm{\xi}}-\bm{\xi})(\hat{\bm{\xi}}-\bm{\xi})\Transpose\big\} \succeq \mathbf{C}\CCRB_{\bm{\xi}}$, where the expectation is with respect to the noise.

Finally, we define \ac{OEB}, \ac{PEB}, \ac{IPEB}, and \ac{SEB}, which show the lower bound on the \ac{RMSE} of estimation as
\begin{alignat}{4}
    \mathrm{OEB} &= \left[\tr\left(\mathbf{C}_{\R_\UE}\right)\right]^{\nicefrac{1}{2}} &, ~
    \mathrm{PEB} &= \left[\tr\left(\mathbf{C}_{\p_\UE}\right)\right]^{\nicefrac{1}{2}} &, \notag \\
    \mathrm{IPEB} &= \ColorBlue{\left[\nicefrac{\textstyle\sum_{m=1}^M \tr\left(\mathbf{C}_{\p_m}\right)}{M}\right]^{\nicefrac{1}{2}}} &, ~
    \mathrm{SEB} &= \left[\tr\left(\mathbf{C}_{b}\right)\right]^{\nicefrac{1}{2}} &, \notag
\end{alignat}
where $\mathbf{C}_{\R_\UE}$, $\mathbf{C}_{\p_\UE}$, $\mathbf{C}_{\p_m}$, and $\mathbf{C}_{b}$ are diagonal sub-matrices in $\mathbf{C}\CCRB_{\bm{\xi}}$ corresponding to $\mathbf{r}$, $\p_\UE$, $\p_m$, and $b$, respectively. We note that $\mathbf{C}_{b}$ is a scalar equal to the variance of clock bias estimation, and that the \ac{RMSE} of estimation of $\mathbf{r}$ is equal to \smash{$\mathbb{E} \left\{ \Vert \mathbf{r} - \hat{\mathbf{r}} \Vert^2 \right\} = \mathbb{E}\big\{\Vert \R_\UE-\hat{\R}_\UE\Vert^2_{\text{F}}\big\}$}. \ColorBlue{In addition, the \ac{IPEB} (in meters) represents the \ac{RMSE} of all the \ac{IP} location estimates, and is a simplified form of the widely used GOSPA metric, from radar sensing and target tracking  \cite{rahmathullah2017generalized}.}

\section{Numerical Results}\label{section:result}
\begin{figure}[t!]
    \centering
    \begin{tikzpicture}
    \node(image)[anchor=south west] {\includegraphics[width=0.8\linewidth]{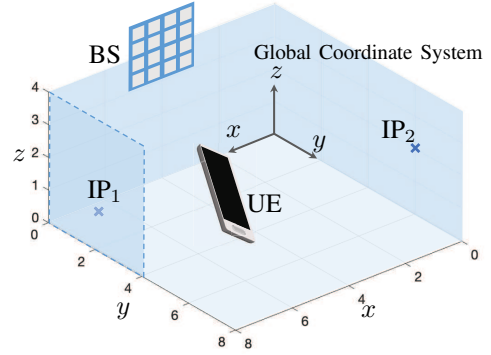}};
    \gettikzxy{(image.north east)}{\ix}{\iy};
    % \draw[help lines] (0,0) grid (\ix,\iy);
    \node at (2.2/9*\ix,5.5/7*\iy){\ac{BS}};
    \node at (4.8/9*\ix,3/7*\iy){\ac{UE}};
    \node at (2.2/9*\ix,3.2/7*\iy){\ac{IP}$_1$};
    \node at (7/9*\ix,4.2/7*\iy){\ac{IP}$_2$};
    \node at (2.5/9*\ix,1.2/7*\iy){$y$};
    \node at (6.5/9*\ix,1.2/7*\iy){$x$};
    \node at (0.8/9*\ix,3.8/7*\iy){$z$};
    \node at (4.3/9*\ix,4.2/7*\iy){$x$};
    \node at (5.7/9*\ix,4/7*\iy){$y$};
    \node at (5/9*\ix,5.2/7*\iy){$z$};
    \node (a) [align=left] at (6.5/9*\ix,5.5/7*\iy){ \footnotesize{Global Coordinate System}};
    \end{tikzpicture}\vspace{-5mm}
    \caption{The indoor scenario considered in simulations, with default parameters provided in Table \ref{table: simulation parameters}.}
    \label{fig: The Indoor Scenario}\vspace{-5mm}
\end{figure}
\subsection{Simulation Setup}
Our simulation scenario consists of an indoor environment shown in Fig.~\ref{fig: The Indoor Scenario}, where the \ac{BS} is mounted vertically. %, at the position given in Table \ref{table: simulation parameters}, and the position as well as the orientation of \ac{UE} is to be estimated. 
We employ \acp{UPA} in both BS and UE, consisting of rectangular configurations of $\sqrt{N_{\BS}} \times \sqrt{N_{\BS}}$ and $\sqrt{N_{\UE}} \times \sqrt{N_{\UE}}$ antennas, with half-wavelength inter-element spacing. Assuming the configurations in the reference orientation where the planar arrays are parallel to the global XY-axes, facilitates expressing $\bm{\Delta}_\BS$ and $\bm{\Delta}_\UE$, i.e., the matrices containing the positions of antenna elements in local coordinate frames\footnote{The antenna element located at row $i$ and column $j$ of such configuration for \ac{BS}, $1 \le i,j \le \sqrt{N_\BS}$, is $(i-1)\sqrt{N_\BS}+j$-th antenna, which is located at $\mathbf{x}_{\BS,(i-1)\sqrt{N_{\BS}}+j}= [j-(\sqrt{N_{\BS}}+1)/2,-i+(\sqrt{N_{\BS}}+1)/2,0]\Transpose d_\BS$, with $d_\BS = \lambda/2$. Similarly, we express the positions of antenna elements in \ac{UE} array.}.
For the channel model, we correspond incidence points to reflecting surfaces with reflection coefficient $\Gamma_{\mathrm{ref}}$ \footnote{For the specific \ac{IP}$_1$ and \ac{IP}$_2$, at the given positions $\p_1$ and $\p_2$, we assume $\Gamma_{\mathrm{ref},1}=0.2$ and $\Gamma_{\mathrm{ref},2}=0.8$, respectively. However, for the general \acp{IP} at random positions, $\Gamma_{\mathrm{ref}}=0.7$.}. We consider that the channel gains are proportional to the free-space path loss, with a random phase uniformly distributed in $[0,2\pi)$, and account for the radiation pattern of antennas \cite[Chapter 4]{balanis2015antenna} as follows:
\begin{align}
|h_{m}|^2 = 
    \begin{cases}
        \hfil \dfrac{\lambda^2\cos
    ^2\theta^{\el}_{\A,0}\cos^2 \theta^{\el}_{\D,0}}{(4\pi)^2 \Vert\p_\BS-\p_\UE\Vert^2} & m=0\\
        \dfrac{\lambda^2\Gamma_{\mathrm{ref}}\cos^2 \theta^{\el}_{\A,m}\cos^2 \theta^{\el}_{\D,m}}{(4\pi)^2 (\Vert\p_\BS-\p_m\Vert+\Vert\p_\UE-\p_m\Vert)^2} & m\neq0.
    \end{cases}\notag
\end{align}

The pilots are set to $x_{k,n}=\sqrt{E_s}$ and the components in precoding and combining vectors $\bm{f}_k$ and $\bm{w}_k$ are assumed to comprise unit-modulus elements with random phase, different for each \ac{OFDM} symbol.

The rotation matrices are also generated with help of Euler angles $\alpha \in [0,2\pi),~\beta \in [0,\pi),~\gamma \in [0,2\pi)$ \cite{TutorialOnSE(3)--Blanco}, using $\R = \R_z(\alpha)\R_y(\beta)\R_x(\gamma)$, where $\R_z(\alpha)$, $\R_y(\beta)$, and $\R_x(\gamma)$ are transformation matrices for counter-clockwise rotations around $z$, $y$, and $x$ axes through $\alpha$, $\beta$, and $\gamma$, respectively \cite{TutorialOnSE(3)--Blanco}. 

\begin{table}%[t!]
\scriptsize
    \centering
    \caption{Default simulation parameters. Parameters that vary are marked with $*$.}
    \label{table: simulation parameters}
    \begin{tabular}{l|c|c}
        \hline\hline 
        Parameter  & Symbol  & Value \\
        \hline
        Propagation Speed & $c$ & $3 \times 10^8$ m/s\\
        Carrier Frequency & $f_c$ & $28$ GHz \\
       % Wavelength & $\lambda$ & $c/f_c$\\
        Subcarrier Spacing & $\Delta_f$ & $120$ kHz \\
        %Bandwidth$^*$ & $B$ & $400$ MHz \\
        $\#$ Subcarriers$^*$ & $N_f$ & $3333$ \\
        $\#$ \ac{OFDM} Symbols & $K$ & $10$\\
        Transmit Power$^*$ & $P_{\mathrm{TX}}$ & $10$ dBm\\
        Noise PSD & $N_0$ & $-174$ dBm/Hz \\
        \ac{UE} Noise Figure & $n_0$ & $13$ dB\\
        \ac{BS} $\#$ Antennas$^*$ & $N_{\BS}$ & $64~(8 \times 8)$\\
        \ac{UE} $\#$ Antennas$^*$ & $N_{\UE}$ & $4 ~(2 \times 2)$\\
        %\ac{BS} Inter-element Spacing & $d_\BS$ & $\lambda/2$\\
        %\ac{UE} Inter-element Spacing & $d_\UE$ & $\lambda/2$\\
        \ac{BS} Position & $\p_{\BS}$ & $[4,0,4]\Transpose$\\
        \ac{BS} Orientation & $\R_{\BS}$ & $\R_x(-\pi/2)$\\
        \ac{UE} Position$^*$ & $\p_{\UE}$ & $[5,4,1]\Transpose$\\
        \ac{UE} Orientation$^*$ & $\R_{\UE}$ & given in \eqref{orientation of UE in simulations}\\
        \ac{IP} Positions$^*$ & $\p_m$ & given in \eqref{positions of IPs in simulations}\\
        Reflection Coefficient$^*$ & $\Gamma_{\mathrm{ref}}$ & $[0.2,0.8]~,~0.7$\\
        Clock Offset & $b$ & $100$ ns\\
        $\#$ Monte Carlo Simulations & $N_s$ & 1000\\
        \hline\hline 
    \end{tabular}
\end{table}

Table \ref{table: simulation parameters} lists all the default simulation parameters, where default orientation for \ac{UE} is either of $\R_1$ or $\R_2$ given by
\begin{align} \label{orientation of UE in simulations}
    \R_1 = \R_z(\nicefrac{\pi}{6})\R_y(\nicefrac{-\pi}{4})\R_x(\nicefrac{-\pi}{36}),\R_2 = \R_x(\nicefrac{\pi}{2}).
\end{align}
In some cases, we generate the orientation randomly. More precisely, $\alpha$, $\beta$, and $\gamma$ are generated randomly. Although we note that this is not equivalent to uniform sampling of $\mathrm{SO}(3)$, it is an easy way to evaluate different orientations. Also, in some simulations, the positions of \acp{IP} are generated randomly, while the positions of default \acp{IP} labeled in Fig.~\ref{fig: The Indoor Scenario} are given in 
% \begin{align} \label{positions of IPs in simulations}
%     \mathbf{P} = [\p_1,\ldots,\p_6] = 
%     \begin{bmatrix}
%       8 & 0 & 3    & 6 & 5.9 & 4    \\
%       2 & 3 & 8    & 6 & 5.9 & 0.05 \\
%       3 & 2 & 1.05 & 0 & 0   & 0
%     \end{bmatrix}.
% \end{align}
\begin{align} \label{positions of IPs in simulations}
    \mathbf{P} = [\p_1,\p_2] = \big[[8,2,1]\Transpose,[0,6,2]\Transpose\big].
\end{align}
The simulations are done in MATLAB $2021$b. For the optimization on manifolds, we utilize the Manopt toolbox \cite{manopt}.
%%%%%%%%%%%%%%%%%%%%%%
\showfigure{
 \begin{figure}%[t!]
     \centering
     \includegraphics[width=\linewidth]{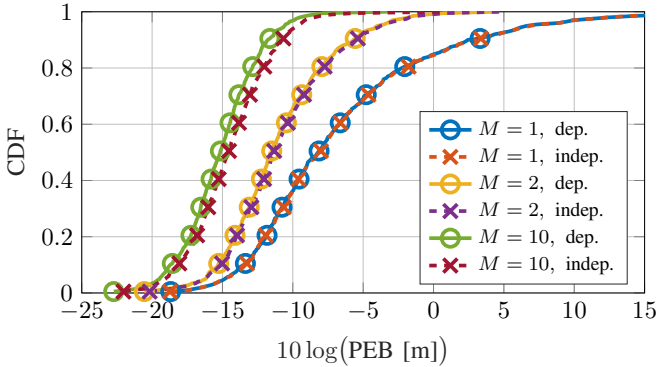}
     \caption{\ac{CDF} of \ac{PEB} with and without independence of channel parameters for a varying number of incidence points.}
     \label{fig: impact_of_dependencies} \vspace{-5mm}
 \end{figure}}
%%%%%%%%%%%%%%%%%%%%%%
\subsection{Obtaining the Likelihood Parameters}
\label{subsection: obtaining independent parameters}
To obtain the parameters $\bm{\kappa}_{\A}$, $\bm{\kappa}_{\D}$, and $\Sigma_{\bm{\tau}}$  presented in Section \ref{sec:Methodology}, considering the independence of angles and delays, we derive the covariance matrix 
\begin{align}
    \mathbf{C}_{\bm{\eta}} = \mathrm{diag}\big(\mathbf{J}_{\bm{\eta}}^{-1}\big) = \mathrm{blkdiag}(\mathbf{C}_{\bm{\theta}_\A},\mathbf{C}_{\bm{\theta}_\D},\mathbf{C}_{\bm{\tau}}),
\end{align}
where $\mathbf{C}_{\bm{\theta}_\A}$, $\mathbf{C}_{\bm{\theta}_\D}$ and $\mathbf{C}_{\bm{\tau}}$ are diagonal matrices corresponding to \acp{AoA}, \acp{AoD}, and \acp{ToA}, respectively. Furthermore, using the independence assumption, \acp{FIM} of $\bm{\theta}_\A$ and $\bm{\theta}_\D$ are given, with respect to $\bm{\kappa}_{\A}$ and $\bm{\kappa}_{\D}$, respectively, by \cite{3DOrientationEstimation--M.Nazari_G.Seco-Granados_P.Johannisson_H.Wymeersch}
\begin{subequations} \label{obtaining kappas and FIM of angles}
\begin{align}
	    \mathbf{J}_{\bm{\theta}_\A}= \mathrm{diag}\big(\bm{\kappa}_\A \odot I_1(\bm{\kappa}_\A) \oslash I_0(\bm{\kappa}_\A) \big),\\
	    \mathbf{J}_{\bm{\theta}_\D}= \mathrm{diag}\big(\bm{\kappa}_\D \odot I_1(\bm{\kappa}_\D) \oslash I_0(\bm{\kappa}_\D) \big),
\end{align}
\end{subequations}
where $I_1(\cdot)$ is the modified Bessel function of order $1$. Solving the above equations gives $\bm{\kappa}_\A$ and $\bm{\kappa}_\D$ \footnote{Specially in high \ac{SNR} regimes, the values in $\bm{\kappa}_\A$ and $\bm{\kappa}_\D$ are large, and therefore, $I_1(\bm{\kappa}_\A) \oslash I_0(\bm{\kappa}_\A) \rightarrow \mathbf{1}_{2(M+1)}$ and $I_1(\bm{\kappa}_\D) \oslash I_0(\bm{\kappa}_\D) \rightarrow \mathbf{1}_{2(M+1)}$, leading to $\bm{\kappa}_\A = \mathrm{diag}\big(\mathbf{C}_{\bm{\theta}_\A}^{-1}\big)$ and $\bm{\kappa}_\D = \mathrm{diag}\big(\mathbf{C}_{\bm{\theta}_\D}^{-1}\big)$. There are other approximations for the ratio $I_1(x) \slash I_0(x)$, for example \cite[Lemma 2]{segura2011bounds}.}. Obviously, $\Sigma_{\bm{\tau}}=\mathbf{C}_{\bm{\tau}}$. 

In order to motivate removing dependencies, we evaluate the \ac{CDF} of \ac{PEB} with and without the independence assumption, using $N_s=1000$ Monte Carlo simulations. In every simulation, we randomize $\R_\UE$ as well as $\p_\UE,~\p_1,\ldots,\p_M \in  [0,8]\times[0,8]\times[0,4]~$m$^3$. Then we obtain the \ac{PEB} in two cases: the general case using $ \mathbf{J}_{\bm{\eta}}$ in \eqref{eq:unconstrained FIM} and then \eqref{eq:CCRBfinal}; the independent case, with  $ \mathbf{J}^{\text{ind}}_{\bm{\eta}} = (\mathrm{diag}\big(\mathbf{J}_{\bm{\eta}}^{-1}\big))^{-1}$ in \eqref{eq:unconstrained FIM} and then \eqref{eq:CCRBfinal}. The \ac{CDF} curves are shown for three different numbers of \acp{IP}, i.e., $M \in \{1,2,10\}$. As seen in Fig.~\ref{fig: impact_of_dependencies}, the distribution of \ac{PEB} with the independence of channel parameters is closely following that of the general case, meaning that not only the angles and delays of different paths but also the azimuth and elevation angles of every individual path can be taken to be independent without a considerable impact. Similar observations hold for the other bounds, i.e., \ac{OEB}, \ac{IPEB}, and \ac{SEB}. Although in certain cases the performance with dependencies can differ significantly (up to $50\%$) from the independent case, the effect is limited on average.

% \begin{figure}%[t!]
% \begin{subfigure}[t]{0.45\textwidth}
%     \centering
%     \includegraphics[width=\linewidth]{results/tikz/performance_adhoc_OEB.tikz}
%   % \caption{}
%     %\label{fig: RMSE of UE orientation}
% \end{subfigure}
% \hfill
% \vspace{4mm}
% \begin{subfigure}[t]{0.45\textwidth}
%     \centering
%     \includegraphics[width=\linewidth]{results/tikz/performance_adhoc_PEB.tikz}
%   % \caption{}
%     %\label{fig: RMSE of UE position}
% \end{subfigure}
% \caption{\ac{RMSE} of \ac{UE} orientation estimation (top) and  \ac{UE} position estimation (bottom) for the ad-hoc estimator, for $M\in \{1,2\}$ and different values of search granularity $\Delta\psi\in \{\pi/10,\pi/50,\pi/100\}$, with the default $\p_\UE$, $\p_1$ (and $\p_\UE$), as well as $\R_\UE=\R_1$. The figures also include the \ac{OEB} and \ac{PEB}.} \label{fig:RMSE}
% \vspace{-4mm}
% \end{figure}
%%%%%%%%%%%%%%%%%%%%%%
\showfigure{
 \begin{figure}
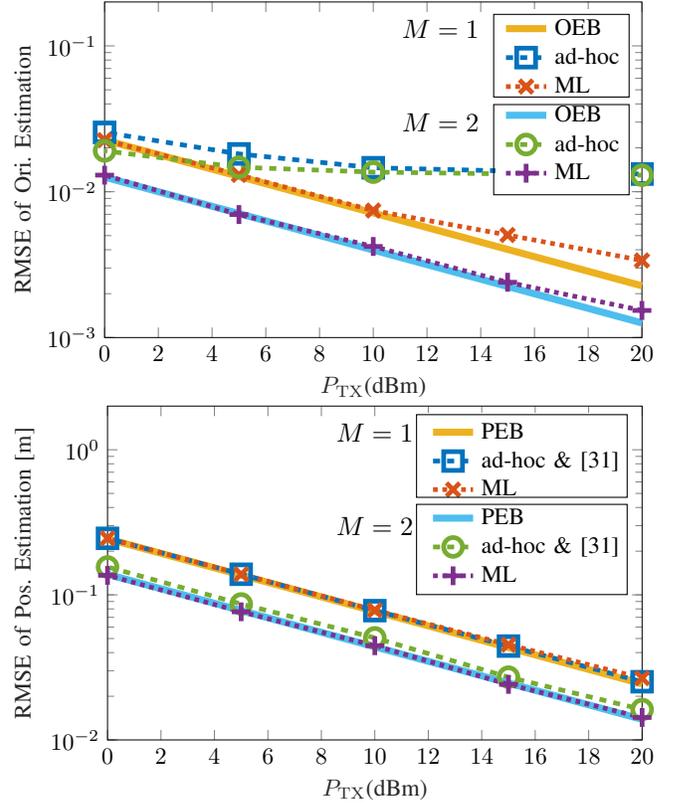
%[t!]
 \begin{subfigure}[t]{0.48\textwidth}
     \centering
     \includegraphics[width=\linewidth]{results/tikz/performance_Orientation.tikz}
   % \caption{}
     %\label{fig: RMSE of UE orientation}
 \end{subfigure}
 \hfill
 \begin{subfigure}[t]{0.48\textwidth}
     \centering
     \includegraphics[width=\linewidth]{results/tikz/performance_Positioning.tikz}
   % \caption{}
     %\label{fig: RMSE of UE position}
 \end{subfigure}
 \caption{\ColorBlue{\ac{RMSE} of \ac{UE} orientation estimation (top) and \ac{UE} position estimation (bottom) vs. $P_{\mathrm{TX}}$, for \ac{ML} and ad-hoc estimators, with search granularity $\pi/200$, for $M\in \{1,2\}$, with $\R_\UE=\R_2$ and the default $\p_\UE$, $\p_1$, and $\p_2$ (in case of $M=2$). The figures also include the \ac{OEB} and \ac{PEB}.}} \label{fig:RMSE}
 \vspace{-4mm}
 \end{figure}}
%%%%%%%%%%%%%%%%%%%%%%
\subsection{Results and Discussions}
\subsubsection{Performance Evaluation of \ac{ML} and Ad-hoc Estimators}
The performance evaluation for the proposed estimators is shown in Fig.~\ref{fig:RMSE}, where we show the \ac{RMSE} of \ac{UE} orientation and position estimation vs. the transmit power $P_{\mathrm{TX}}$, for two cases, i.e., $M=1$ with \ac{IP}$_1$, and $M=2$, with both \ac{IP}$_1$ and \ac{IP}$_2$, with the reflection coefficients $\Gamma_{\mathrm{ref}}=0.2$ for \ac{IP}$_1$, and $\Gamma_{\mathrm{ref}}=0.8$ for \ac{IP}$_2$. We observe that the performance of both estimators improves by increasing the transmit power, closely following the corresponding bounds. Specifically, the \ac{RMSE} of \ac{UE} position estimation using the proposed ad-hoc routine sees a negligible gap compared to the \ac{CRB}, for a large range of transmit powers. This of course depends on the geometry as well as the granularity of the 1-dimensional search for obtaining $\psi$. Not surprisingly, in the low \ac{SNR} regimes, the performance deviates from the bound, but in the case of $M=2$, the \ac{ML} estimator is able to reduce the gap. This gap is due to the ignorance of the distribution of measurements in the ad-hoc estimator, and it is especially pronounced in unfavorable positions of \ac{UE}, where the quality of different paths arriving at \ac{UE} are remarkably different. In moderate \ac{SNR}, the \ac{ML} estimation %tightens the possible gaps, and gives the
yields an \acp{RMSE} close to the performance bounds. In very high \ac{SNR} regimes, the performance of the ad-hoc estimator saturates due to the granularity of angle search. It is then refined using the \ac{ML}, and the gap to the bounds is substantially reduced. %, while it might not exactly touch the bounds.  
\ColorBlue{As a benchmark, we compare the positioning performance of the ad-hoc estimator to that of the only other available estimator, i.e., \cite[eq.~(15)--(16)]{Nuria}. We observe that the  proposed ad-hoc estimator achieves identical performance. The strength of \cite{Nuria} is that it does not require an estimate of the rotation matrix for estimating the \ac{UE} position, but at the same time, the one-dimensional search of our method for estimation of rotation matrix should be done anyway, to initialize the \ac{ML} algorithm.} 
The tightness of \ac{ML} estimator to the \ac{CRB}, and the negligible gap between the performance of the ad-hoc estimator and the lower bounds for a practical range, shows the efficiency of our proposed estimation algorithms. %Obviously, when two of the incidence points are present, a better performance is achieved.

%%%%%%%%%%%%%%%%%%%%%%
\showfigure{
 \begin{figure*}
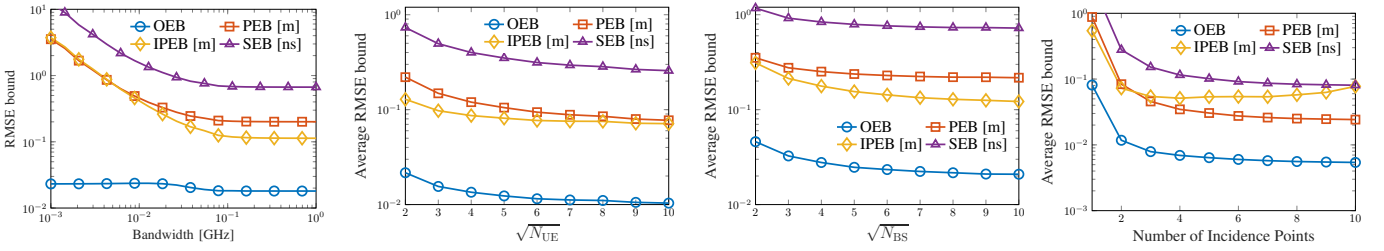
%[t!]
 \begin{subfigure}[t]{0.24\textwidth}
     \centering
     \includegraphics[width=\linewidth]{results/tikz/CRB_vs_BW.tikz}
 %    \caption{}
  %   \label{fig: impact_of_BW}
 \end{subfigure}
 \hfill
 \begin{subfigure}[t]{0.24\textwidth}
     \centering
     \includegraphics[width=\linewidth]{results/tikz/CRB_vs_N_UE.tikz}
   % \caption{}
  %   \label{fig: impact_of_N_UE}
 \end{subfigure}
 \hfill
 \begin{subfigure}[t]{0.24\textwidth}
     \centering
     \includegraphics[width=\linewidth]{results/tikz/CRB_vs_N_BS.tikz}
   %  \caption{}
 %    \label{fig: impact_of_N_BS}
 \end{subfigure}
 \begin{subfigure}[b]{0.24\textwidth}
     \centering
     \includegraphics[width=\linewidth]{results/tikz/CRB_vs_num_NL.tikz}
   %\caption{}
   % \label{fig: number of NLoS paths}
 \end{subfigure}
 \caption{\ColorBlue{Impact of bandwidth (left),  number of \ac{UE} antennas (middle-left),  number of \ac{BS} antennas (middle-right), and number of \acp{IP} (right) on (average) performance error bounds.}} \label{fig:impactOfParameters} \vspace{-5mm}
 \end{figure*}}
%%%%%%%%%%%%%%%%%%%%%%
\subsubsection{Impact of System Parameters}
In Fig.~\ref{fig:impactOfParameters}, we evaluate the impact of bandwidth, the number of antennas, and the number of \acp{IP}, using $\R_\UE=\R_1$ and the default $\p_\UE$. For evaluation of the impact of bandwidth and number of antennas, we consider one \ac{IP} at the default position $\p_1$, while we evaluate the impact of the number of \acp{IP}, in an average sense, i.e., the positions of \acp{IP} are randomized, and the average error bounds are obtained. 

As it is observed in the left plot in Fig.~\ref{fig:impactOfParameters}, increasing the bandwidth, which leads to higher \ac{ToA} accuracy and improved delay resolution, decreases the error bounds. This trend, however, saturates at some point ($\approx 100$ MHz), because further improvement is limited by the accuracy of angle measurements. The \ac{OEB} is the least benefited from the enhancement of \ac{ToA} accuracy, which makes sense since the orientation is determined mainly by angle measurements and not the delays. 

The two middle plots of  Fig.~\ref{fig:impactOfParameters},  show the performance gains achieved when the 
%On the other hand, gains when the 
angular resolution and accuracy improve %, are seen in the , showing the impact of 
thanks to an increase in the number of \ac{UE} and \ac{BS} antennas. \ac{PEB} and \ac{SEB} benefit most from additional \ac{UE} antennas, while \ac{IPEB} benefits most from additional \ac{BS} antennas. 
Since analog combining is used with a fixed number of precoders and combiners, there is no array gain, which leads again to saturation at a larger number of antennas, when there are no further resolution gains, and the performance is limited by the bandwidth.  
%It is worth mentioning that the impact of number of antennas is also evaluated on average, i.e., we carry out Monte Carlo simulations with the codebook as random parameter, and then we obtain the \ac{CRB}. This means that the error bounds improve by increasing number of antennas, on average. The reason why this is not guaranteed for each realization is the fact that precoding and combining vectors are random, leading to beamed signals in random directions with different power. 

Finally, in the right plot of Fig.~\ref{fig:impactOfParameters} we show the impact of the number of \acp{IP}. We see that 
increasing number of \acp{IP} leads to improvements in the \ac{OEB}, \ac{PEB}, and \ac{SEB}. The reduction of error bounds is especially considerable when the number of \acp{IP} changes from $1$ to $2$. The reason is that, when $M=1$, the quality of estimating \ac{AoA} or \ac{AoD} degrades in certain  positions, and this causes larger error bounds, on average. However, when another \ac{IP} is added, the probability of having both \acp{IP} at unfavorable positions reduces significantly, and the average error bounds decrease. For the \ac{IPEB}, while it decreases at the beginning with the number of incidence points, it may experience small fluctuations or even increase at larger $M$. This is due to the increase in the number of unknowns, i.e., positions of \acp{IP}. 

\showfigure{
 \begin{figure*}[t!]
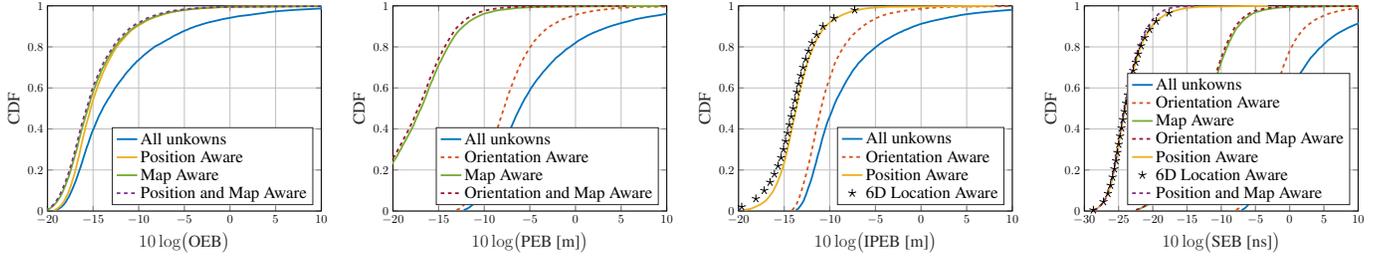

 \vspace{5mm}
 \begin{subfigure}[t]{0.24\textwidth}
     \centering
     \includegraphics[width=\linewidth]{results/tikz/OEB_sideInfo.tikz}
     %\caption{}
   % \label{fig: CDF of OEB - impact of side info}
 \end{subfigure}
 \hfill
 \begin{subfigure}[t]{0.24\textwidth}
     \centering
     \includegraphics[width=\linewidth]{results/tikz/PEB_sideInfo.tikz}
     %\caption{}
   % \label{fig: CDF of PEB - impact of side info}
 \end{subfigure}
\hfill
 \begin{subfigure}[t]{0.24\textwidth}
     \centering
     \includegraphics[width=\linewidth]{results/tikz/MEB_sideInfo.tikz}
    %  \caption{}    
    \label{fig: CDF of IPEB - impact of side info}
 \end{subfigure}
 \hfill
 \begin{subfigure}[t]{0.24\textwidth}
     \centering
     \includegraphics[width=\linewidth]{results/tikz/SEB_sideInfo.tikz}
     %\caption{}
     %\label{fig: CDF of SEB - impact of side info}
 \end{subfigure}
 \caption{\ac{CDF} of \ac{OEB} (left), \ac{PEB} (middle-left), \ac{IPEB} (middle-right), and  \ac{SEB} (right), under different levels of side-information. The position and orientation of \ac{UE} are randomized, and the default \ac{IP}$_1$ is included.} \label{fig:impactOfKnownParameters} \vspace{-5mm}
 \end{figure*}}
%%%%%%%%%%%%%%%%%%%%%%
\subsubsection{Impact of Known Parameters}
In this part, we assess the impact of known parameters, i.e., we evaluate the best achievable estimation accuracy, if some of the parameters are known. For that, we depict the \ac{CDF} of error bounds in Fig.~\ref{fig:impactOfKnownParameters}. To set up the Monte Carlo simulations, we consider only one \ac{IP} at the default position $\p_1$, while we randomize $[p_{\UE,x}, p_{\UE,y}] \in [0,8] \times [0,8]$ in the $p_{\UE,z}=1$ plane (though $p_{\UE,z}=1$  is considered unknown), as well as  the \ac{UE} orientation. The \ac{CDF} curves are shown for $N_s=10,000$ realizations. 
In terms of the \ac{OEB}, 
%As observed in Fig.~\ref{fig: CDF of OEB - impact of side info},
position knowledge of either the \ac{UE} or \ac{IP} improves the orientation accuracy, congruent with the findings from \cite{3DOrientationEstimation--M.Nazari_G.Seco-Granados_P.Johannisson_H.Wymeersch} with $2$ \acp{BS}. In terms of the \ac{PEB}, 
%
%Evaluation of the \ac{CDF} of \ac{PEB}, however, indicates that the impact of 
orientation awareness is less important than the knowledge from the mapping, i.e., the position of the incidence point. %, which escalates the performance by $\sim 7$ dB (see Fig.~\ref{fig: CDF of PEB - impact of side info}). 
Similarly, the knowledge of \ac{UE} position can help mapping, and certainly, if both $\p_\UE$ are $\R_\UE$ are known, lower \ac{IPEB} is achieved. %It seems that the impact of \ac{UE} orientation knowledge on mapping is slightly more than that of positioning, which is expected, because of connection of orientation and the reflected path through $\p_1$.
Finally, in terms of the \ac{SEB}, a variety of cases exist,
%More interesting is the observation of the \ac{CDF} of \ac{SEB}, 
with and without side information from either or some of $\R_\UE$, $\p_\UE$, and $\p_1$. As we have seen on the other bounds, knowledge of orientation is the least informative, and the UE and IP location awareness provide a large amount of information on the clock bias. 
%In contrast, \ac{UE} and \ac{IP} location information provide independently a large amount of information from the clock bias. %and knowledge of position is the most valuable information, if required to choose one side information, with the effectiveness of map-awareness in second place. 
%Including orientation as the known parameter, improves the aforementioned performances slightly, while fusing the mapping information with the knowledge of position, has the greatest impact. 
\showfigure{
\begin{figure}[t!]
\begin{subfigure}[t]{0.48\linewidth}
    \centering
    \begin{tikzpicture}
    \node(image)[anchor=south west] {\includegraphics[trim={0.5cm 0.25cm 0.5cm 0.25cm},clip,width=\linewidth]{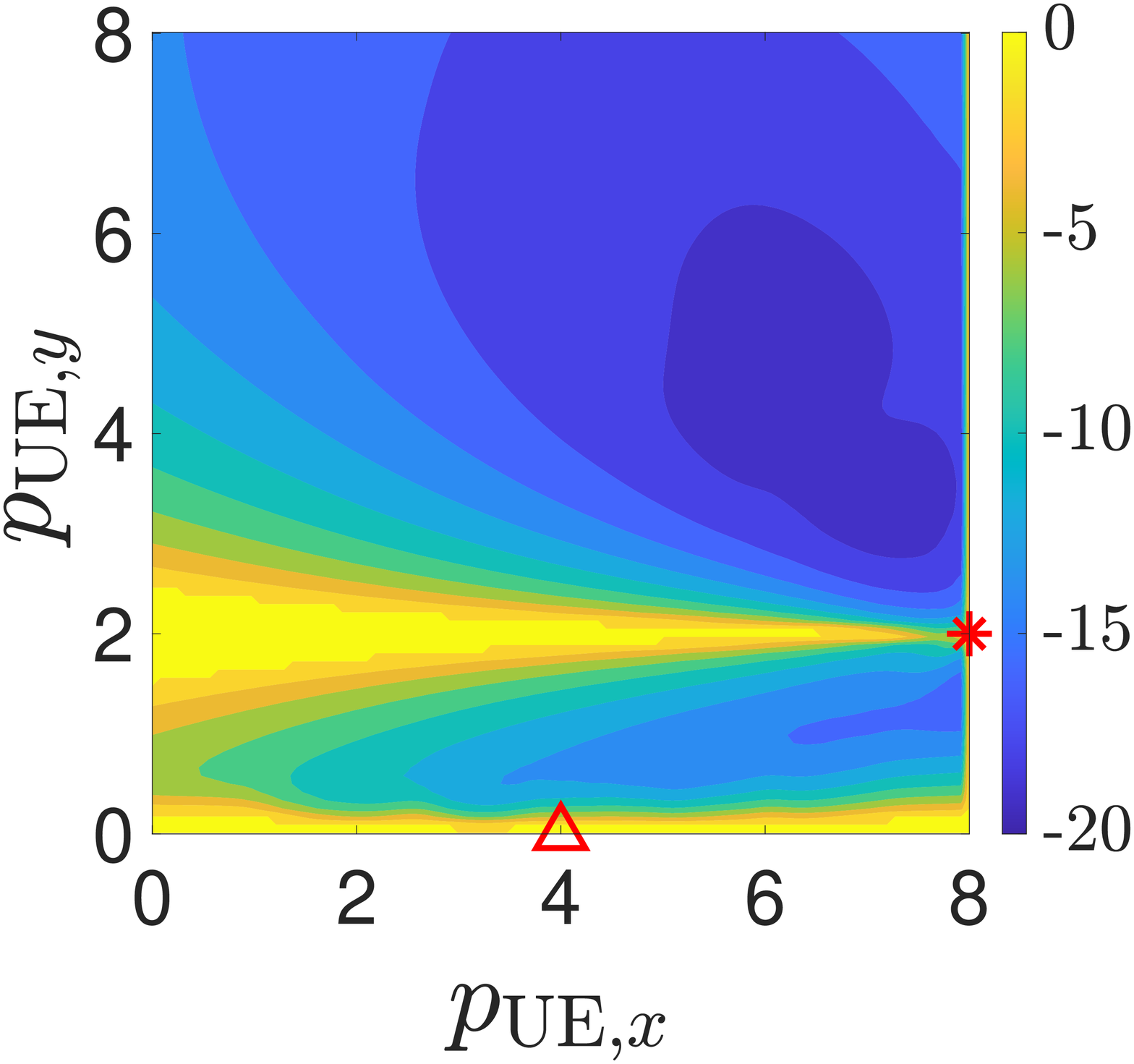}};
    \gettikzxy{(image.north east)}{\ix}{\iy};
   %  \draw[help lines] (0,0) grid (\ix,\iy);
    %\node at (2/4.5*\ix,4.5/4.5*\iy){$10 \log (\text{OEB})$};
    \node at (2.3/4.5*\ix,1.3/4.5*\iy){\color{red}\ac{BS}};
    \node at (3.5/4.5*\ix,1.7/4.5*\iy){\color{red}\ac{IP}$_1$};
    \end{tikzpicture}
    \caption{$10 \log (\text{OEB})$}
    \label{a}
\end{subfigure}
\hfill
\begin{subfigure}[t]{0.48\linewidth}
    \centering
    \begin{tikzpicture}
    \node(image)[anchor=south west] {\includegraphics[trim={0.5cm 0.25cm 0.5cm 0.25cm},clip,width=\linewidth]{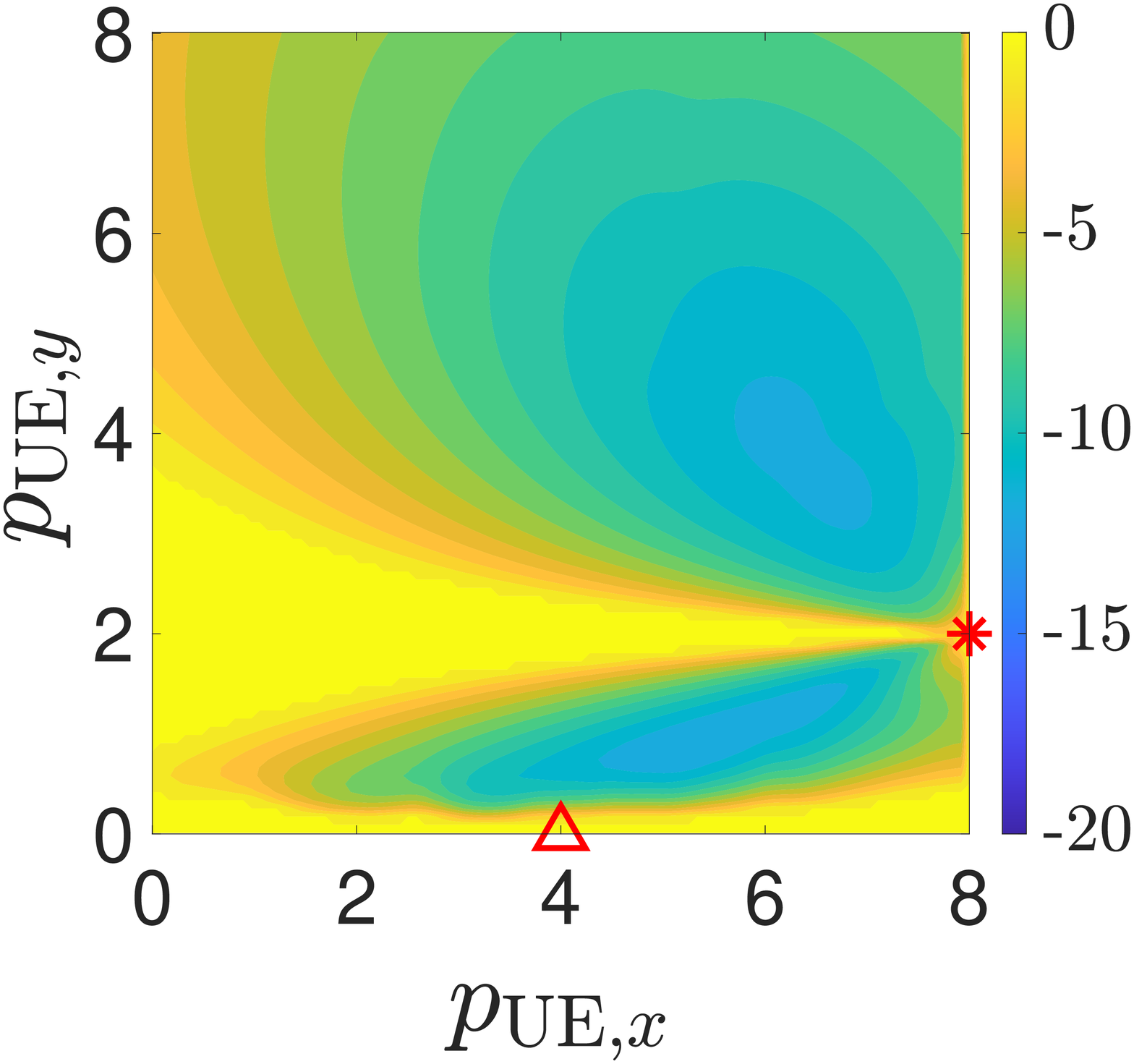}};
    \gettikzxy{(image.north east)}{\ix}{\iy};
    % \draw[help lines] (0,0) grid (\ix,\iy);
    \node at (2.3/4.5*\ix,1.3/4.5*\iy){\color{red}\ac{BS}};
    \node at (3.5/4.5*\ix,1.7/4.5*\iy){\color{red}\ac{IP}$_1$};
    \end{tikzpicture}
    \caption{$10 \log (\text{PEB})$}
    \label{b}
\end{subfigure}
\hfill
\begin{subfigure}[t]{0.48\linewidth}
    \centering
    \begin{tikzpicture}
    \node(image)[anchor=south west] {\includegraphics[trim={0.5cm 0.25cm 0.5cm 0.25cm},clip,width=\linewidth]{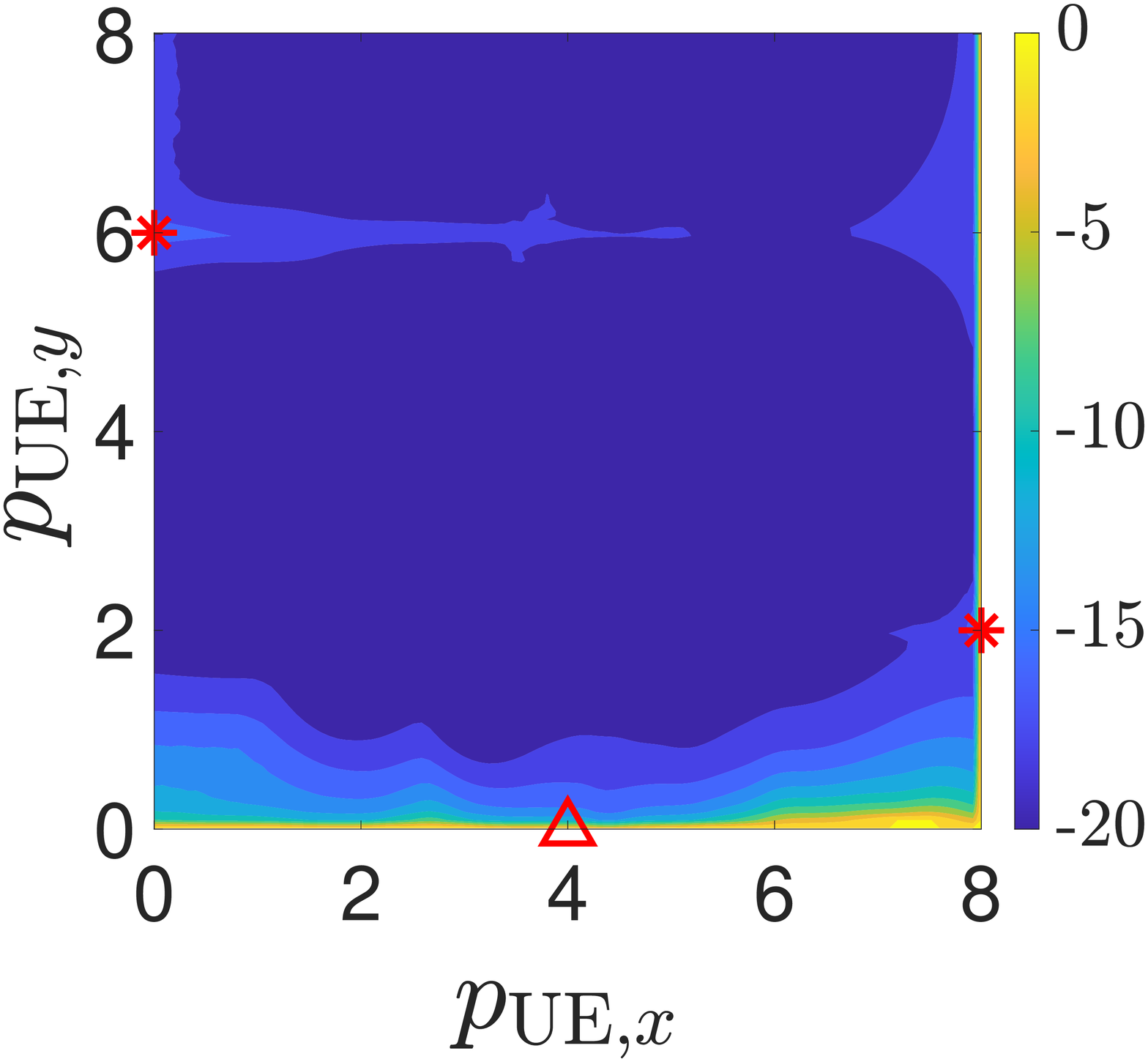}};
    \gettikzxy{(image.north east)}{\ix}{\iy};
    % \draw[help lines] (0,0) grid (\ix,\iy);
    \node at (2.3/4.5*\ix,1.3/4.5*\iy){\color{red}\ac{BS}};
    \node at (3.5/4.5*\ix,1.7/4.5*\iy){\color{red}\ac{IP}$_1$};
    \node at (1.1/4.5*\ix,3.4/4.5*\iy){\color{red}\ac{IP}$_2$};
    \end{tikzpicture}
    \caption{$10 \log (\text{OEB})$}
    \label{aa}
\end{subfigure}
\hfill
\begin{subfigure}[t]{0.48\linewidth}
    \centering
    \begin{tikzpicture}
    \node(image)[anchor=south west] {\includegraphics[trim={0.5cm 0.25cm 0.5cm 0.25cm},clip,width=\linewidth]{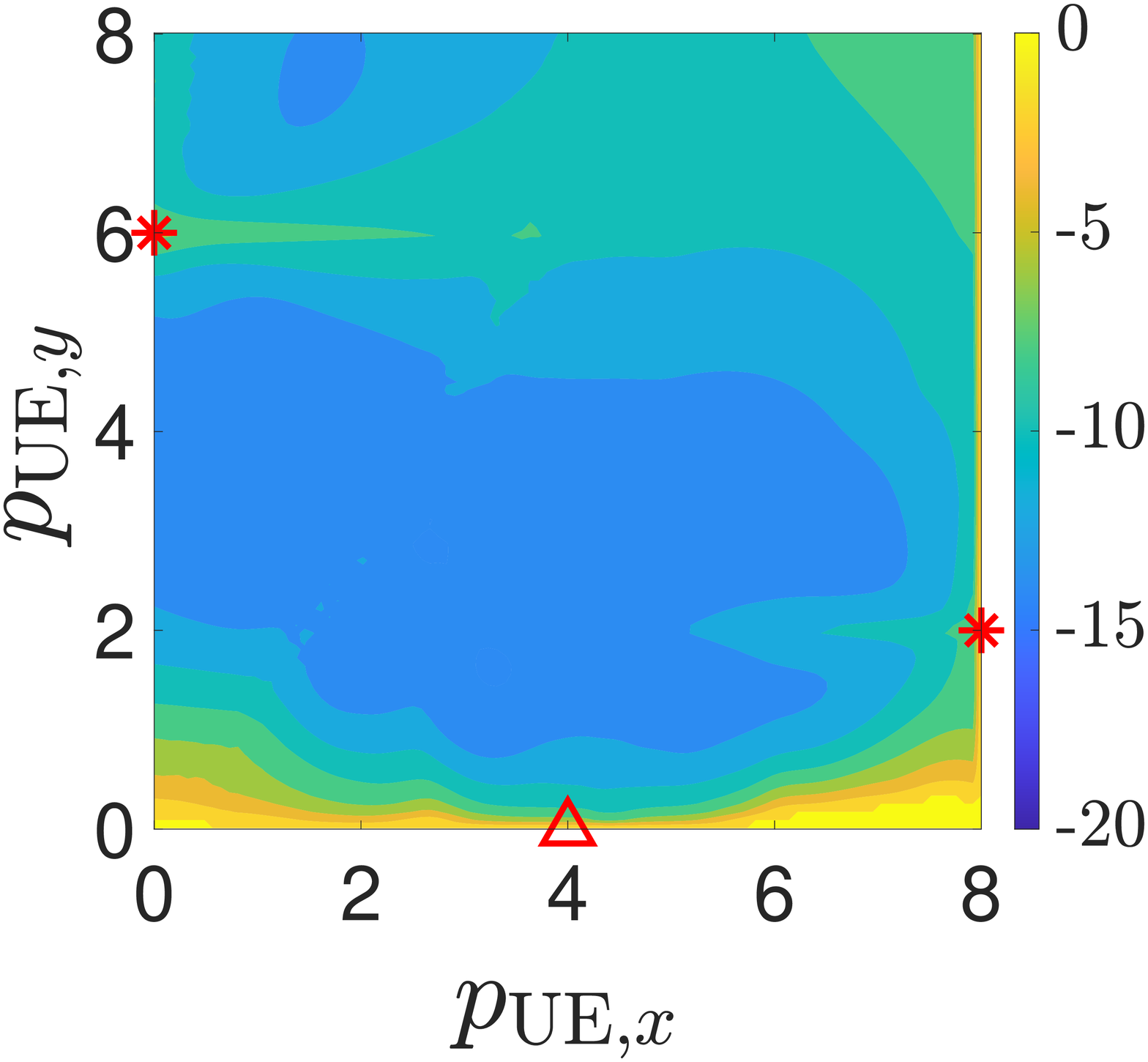}};
    \gettikzxy{(image.north east)}{\ix}{\iy};
    % \draw[help lines] (0,0) grid (\ix,\iy);
    \node at (2.3/4.5*\ix,1.3/4.5*\iy){\color{red}\ac{BS}};
    \node at (3.5/4.5*\ix,1.7/4.5*\iy){\color{red}\ac{IP}$_1$};
    \node at (1.1/4.5*\ix,3.4/4.5*\iy){\color{red}\ac{IP}$_2$};
    \end{tikzpicture}
    \caption{$10 \log (\text{PEB})$}
    \label{bb}
\end{subfigure}
% \begin{subfigure}[t]{0.48\linewidth}
%     \centering
%     \begin{tikzpicture}
%     \node(image)[anchor=south west] {\includegraphics[trim={0.5cm 0.25cm 0.5cm 0.25cm},clip,width=\linewidth]{results/eps/contour_p_UE_IP_1_IPEB.eps}};
%     \gettikzxy{(image.north east)}{\ix}{\iy};
%     % \draw[help lines] (0,0) grid (\ix,\iy);
%     \node at (2.3/4.5*\ix,1.3/4.5*\iy){\color{red}\ac{BS}};
%     \node at (3.5/4.5*\ix,1.7/4.5*\iy){\color{red}\ac{IP}$_1$};
%     \end{tikzpicture}
%     \caption{}
%     \label{c}
% \end{subfigure}
% \hfill
% \begin{subfigure}[t]{0.48\linewidth}
%     \centering
%     \begin{tikzpicture}
%     \node(image)[anchor=south west] {\includegraphics[trim={0.5cm 0.25cm 0.5cm 0.25cm},clip,width=\linewidth]{results/eps/contour_p_UE_IP_1_SEB.eps}};
%     \gettikzxy{(image.north east)}{\ix}{\iy};
%     % \draw[help lines] (0,0) grid (\ix,\iy);
%     \node at (2.3/4.5*\ix,1.3/4.5*\iy){\color{red}\ac{BS}};
%     \node at (3.5/4.5*\ix,1.7/4.5*\iy){\color{red}\ac{IP}$_1$};
%     \end{tikzpicture}
%     \caption{}
%     \label{d}
% \end{subfigure}
\caption{Contour plots of (a) \ac{OEB}, (b) \ac{PEB} [m] with \ac{IP}$_1$; (c) \ac{OEB}, (d) \ac{PEB} [m] with \ac{IP}$_1$ and  \ac{IP}$_2$, for $p_{\UE,z}=1$, $\R_\UE=\R_2$.
%
%, (c) \ac{IPEB} [m], and (d) \ac{SEB} [ns], with values larger than $1$ truncated, for $p_{\UE,z}=1$, $\R_\UE=\R_2$, as well as \ac{IP}$_1$ at default position.
}
\label{fig: contourplot for p_UE with 1 IP}\vspace{-5mm}
\end{figure}}

\showfigure{
\begin{figure}[t!]
\begin{subfigure}[t]{0.48\linewidth}
    \centering
    \begin{tikzpicture}
    \node(image)[anchor=south west] {
    \includegraphics[trim={0.5cm 0.25cm 0.5cm 0.25cm},clip,width=\linewidth]{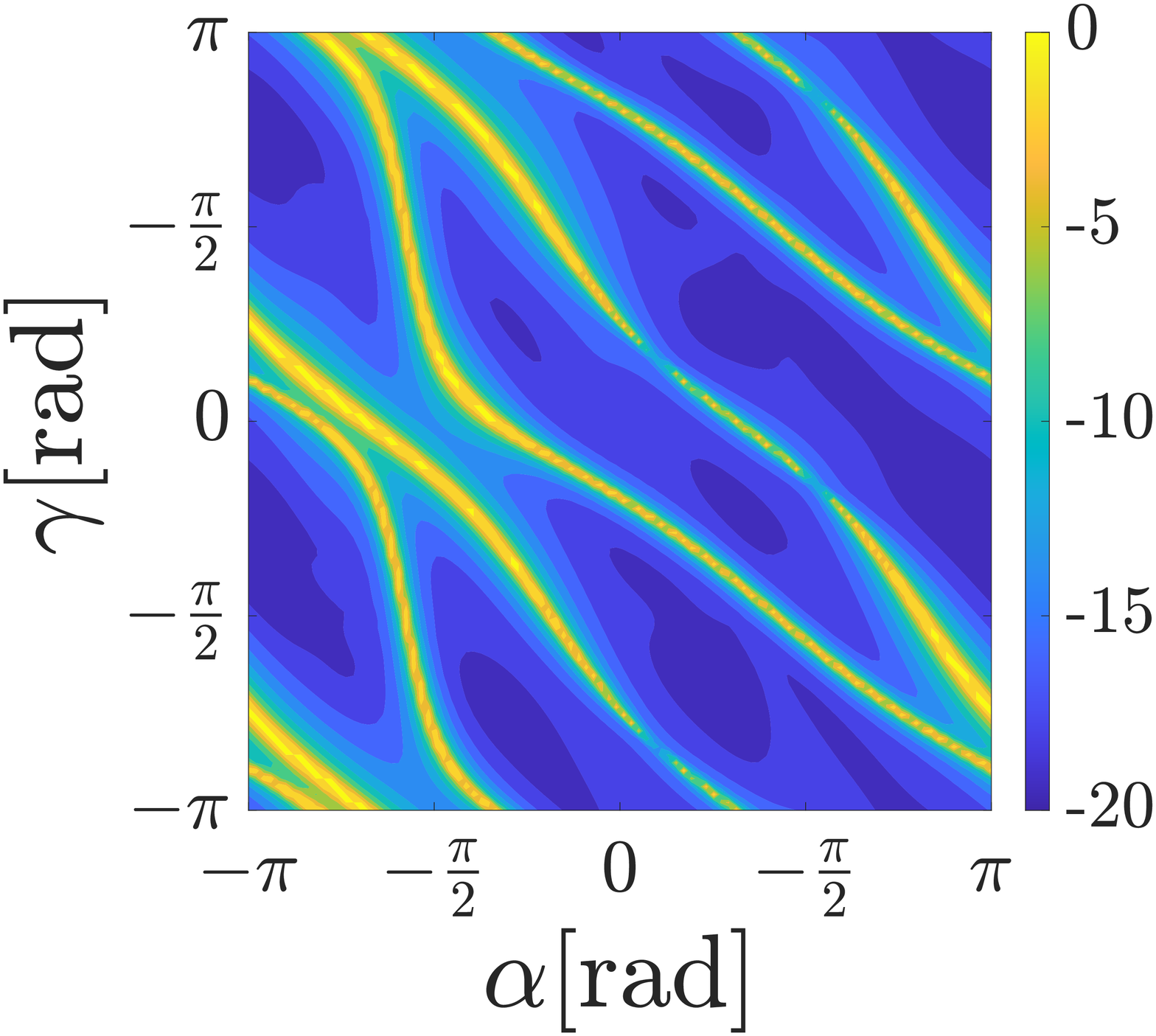}};
    \end{tikzpicture}
    \caption{$10 \log (\text{OEB})$}
    \label{aaa}
\end{subfigure}
\hfill
\begin{subfigure}[t]{0.48\linewidth}
    \centering
    \begin{tikzpicture}
    \node(image)[anchor=south west] {
    \includegraphics[trim={0.5cm 0.25cm 0.5cm 0.25cm},clip,width=\linewidth]{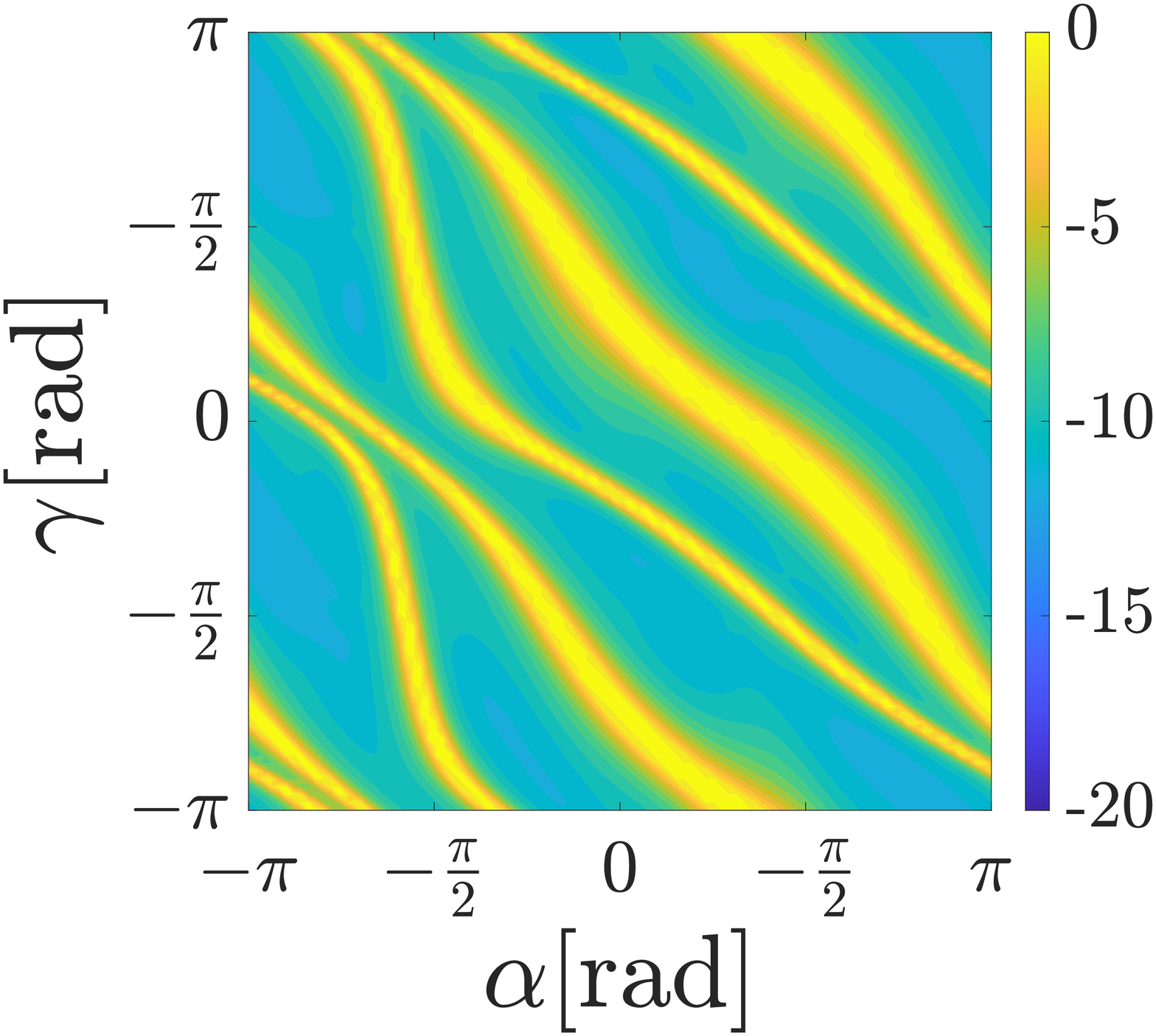}};
    \end{tikzpicture}
    \caption{$10 \log (\text{PEB})$}
    \label{bbb}
\end{subfigure}
\hfill
\begin{subfigure}[t]{0.48\linewidth}
    \centering
    \begin{tikzpicture}
    \node(image)[anchor=south west] {
    \includegraphics[trim={0.5cm 0.25cm 0.5cm 0.25cm},clip,width=\linewidth]{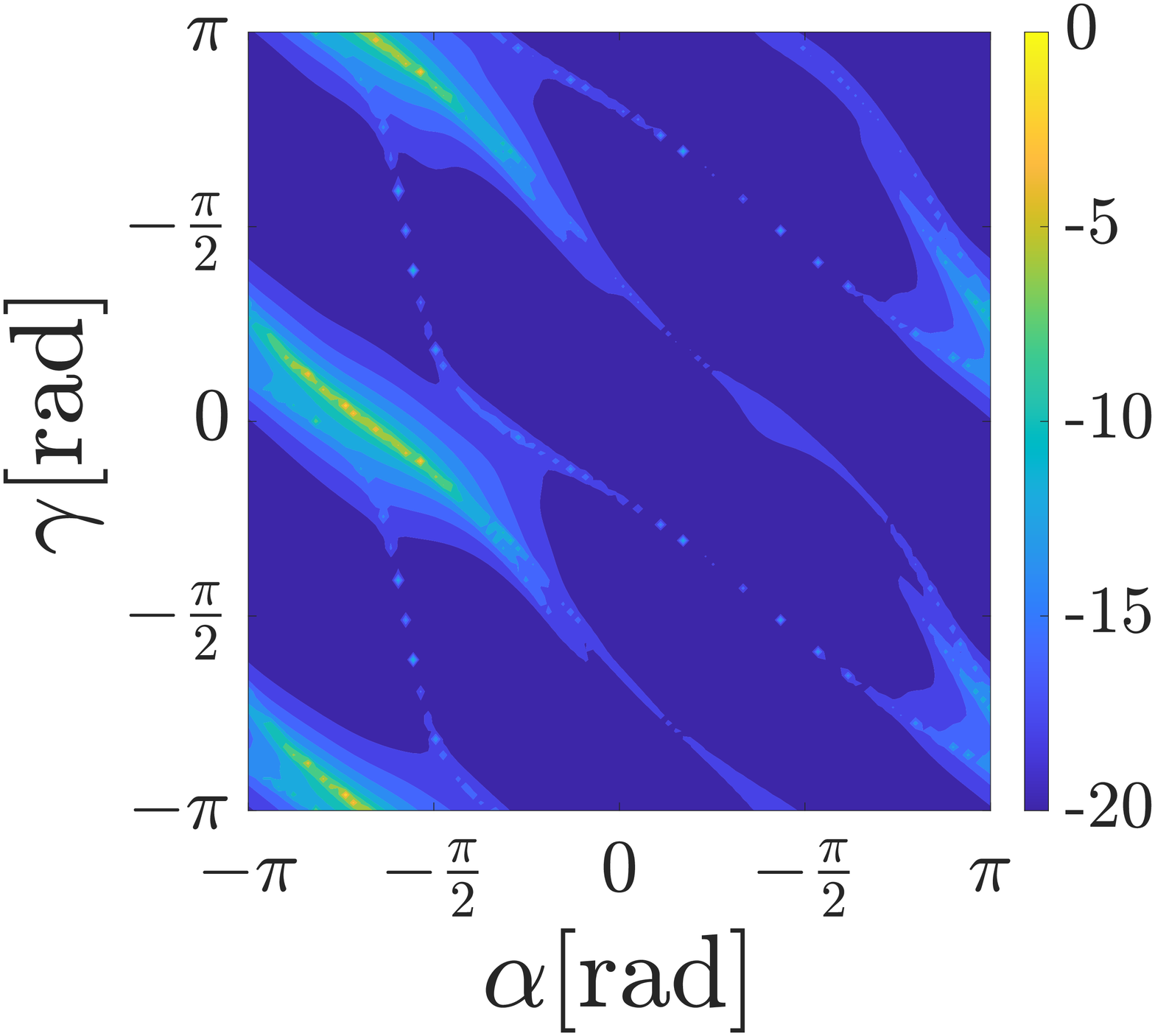}};
    \end{tikzpicture}
    \caption{$10 \log (\text{OEB})$}
    \label{aaaa}
\end{subfigure}
\hfill
\begin{subfigure}[t]{0.48\linewidth}
    \centering
    \begin{tikzpicture}
    \node(image)[anchor=south west] {
    \includegraphics[trim={0.5cm 0.25cm 0.5cm 0.25cm},clip,width=\linewidth]{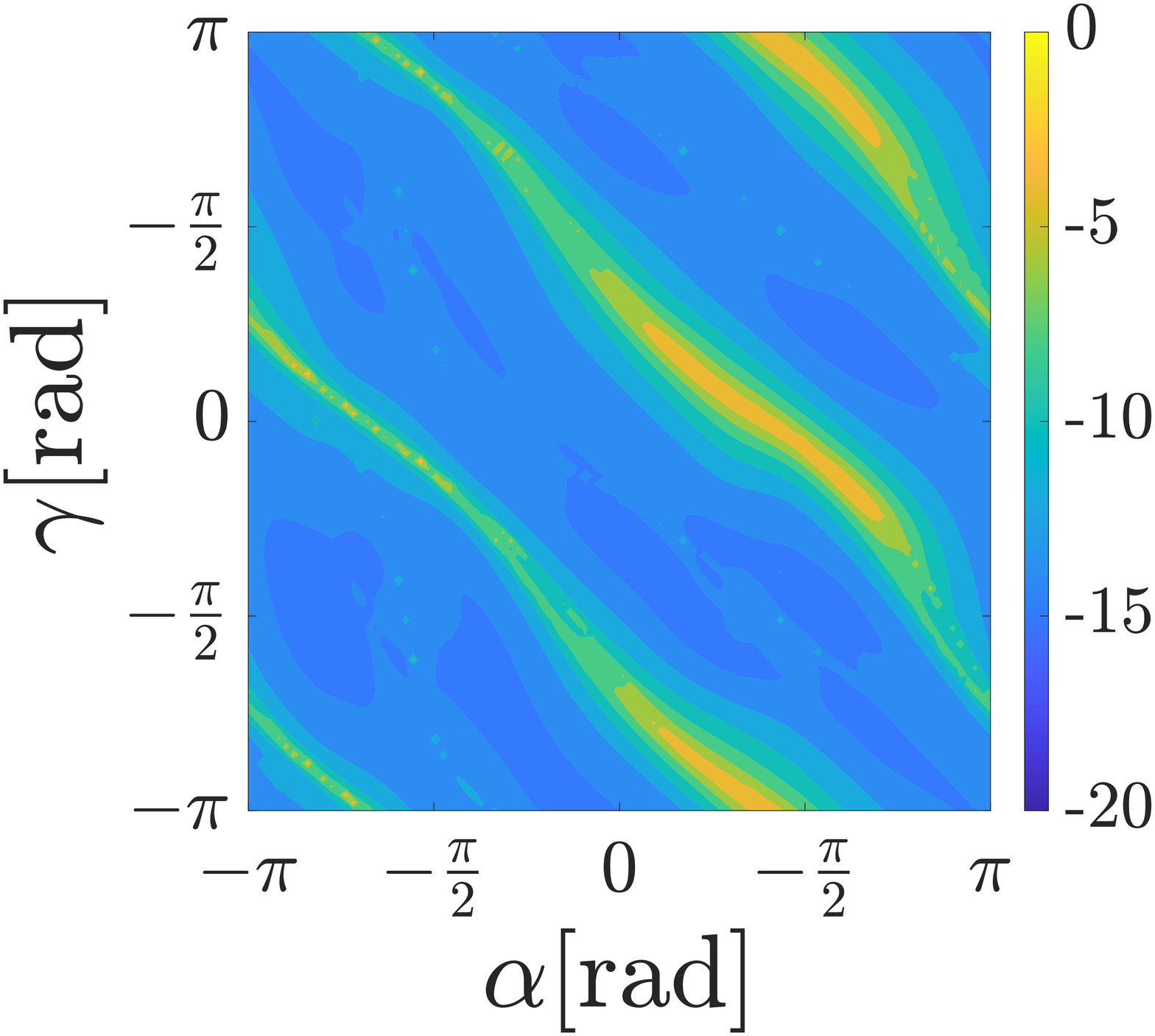}};
    \end{tikzpicture}
    \caption{$10 \log (\text{PEB})$}
    \label{bbbb}
\end{subfigure}
% \begin{subfigure}[t]{0.48\linewidth}
%     \centering
%     \begin{tikzpicture}
%     \node(image)[anchor=south west] {
%     \includegraphics[trim={0.5cm 0.25cm 0.5cm 0.25cm},clip,width=\linewidth]{results/eps/contour_R_UE_IP_1_IPEB.eps}};
%     \end{tikzpicture}
%     \caption{}
%     \label{ccc}
% \end{subfigure}
% \hfill
% \begin{subfigure}[t]{0.48\linewidth}
%     \centering
%     \begin{tikzpicture}
%     \node(image)[anchor=south west] {
%     \includegraphics[trim={0.5cm 0.25cm 0.5cm 0.25cm},clip,width=\linewidth]{results/eps/contour_R_UE_IP_1_SEB.eps}};
%     \end{tikzpicture}
%     \caption{}
%     \label{ddd}
% \end{subfigure}
\caption{Contour plots of of (a) \ac{OEB}, (b) \ac{PEB} [m] with \ac{IP}$_1$; (c)  \ac{OEB}, (d) \ac{PEB} [m] with \ac{IP}$_1$ and \ac{IP}$_2$, 
%
%, (c) \ac{IPEB} [m], and (d) \ac{SEB} [ns], with values larger than $1$ truncated, 
for $\beta=-\pi/4$.}
\label{fig: contourplot for R_UE with 1 IP} \vspace{-5mm}
\end{figure}}

\subsubsection{6D Localization Coverage}
As the last step in our simulations, we evaluate the localization coverage and performance robustness, via contour plots of \ac{OEB} and \ac{PEB}, in the region around the BS and IPs. %and position of \ac{IP}. 
Fig.~\ref{fig: contourplot for p_UE with 1 IP} shows the contour plots of error bounds, when $x$ and $y$ coordinates of \ac{UE} position are varied, while the $z$ coordinate is fixed to $1$, considering  $1$ \ac{NLoS} path (top row) or $2$ \ac{NLoS} paths (bottom row). As the \ac{UE} gets closer to the \ac{IP} (which has the lower channel gain compared to the direct path from the \ac{BS}), the quality of estimation of all parameters improves, unless if \ac{UE} approaches the $y=0$ plane resulting in $\theta_{\D,0}^{\el}$ close to $\pi/2$, which in turn strongly attenuates the \ac{LoS} path, degrades the estimation of $\bm{\theta}_{\D,0}$, and accordingly increases the error. %As the \ac{UE} is at height $z=1$ with a rotation around $z$ axis, the \ac{AoA} estimations are desirable. 
Including \ac{IP}$_2$ at the default position provides another signal source and improves the coverage. In summary, good performance is achieved close to \acp{IP}, with a graceful degradation further away. However, if the \ac{UE} should not be so close to the \ac{IP} that the \ac{NLoS} paths is no longer resolvable from the \ac{LoS} path. 
%range with more satisfactory estimation accuracy, as seen in Fig.~\ref{fig: contourplot for p_UE with 2 IP}.   

In a similar fashion, we depict the contour plots of error bounds, for a range of rotation angles of \ac{UE} in Fig.~\ref{fig: contourplot for R_UE with 1 IP}, by fixing one of the Euler angles $\beta=-\pi/4$. 
If only \ac{IP}$_1$ is present, we observe a continuous set of orientations for which the bounds are infinite. 
%some peaks in the performance error bounds, where they tend to infinity (see Fig.~\ref{fig: contourplot for R_UE with 1 IP}). 
Specifically, the orientation caused by the composition of rotations around $z$ and $x$ axes through $\alpha$ and $\gamma$ respectively, results in the received ray from either \ac{LoS} or \ac{NLoS} paths to hit the \ac{UE} antenna array on the endfire, and does not provide a high-quality estimation of either of \acp{AoA}. Subsequently, both \ac{OEB} and \ac{PEB} are affected. However, once these specific orientations change, the signal arrives in more appropriate directions, and the bounds improve.
Once a second \ac{IP} is added, the problem is non-identifiable for only a reduced set of configurations. Hence, with $M>1$ \ac{NLoS} paths, a more uniform 6D localization coverage can be achieved. 
 %It is also expected that the scenario with inclusion of both \ac{IP}$_1$ and \ac{IP}$_2$ eliminates the peaks and greatly reduces the unbounded error bounds (see Fig.~\ref{fig: contourplot for R_UE with 2 IP}).

While not shown, the performance as a function of the \ac{IP} location can also be evaluated. With a single \ac{IP}, the bounds become infinite when the \ac{IP} is on the 3D line between \ac{BS} and \ac{UE}, which is a very unlikely configuration. 
%Visualization of the achievable performance for different positions of incidence point is also advantageous. As seen in Fig.~\ref{fig: contourplot for p_IP}, where we fix the \ac{UE} at the default position with rotation around $z$ axis by $\pi/6$, and change the $x$ and $y$ coordinates of a single \ac{IP} locating in $z=3$ plane, the problem is often identifiable, unless $p_{\IP,y} \approx 0$ or $\p_{\IP} \approx [4.3,1.3,3]\Transpose$, where in the former, the estimation of $\bm{\theta}_{\D,1}$ is inferior, while in the latter, the \ac{IP} locates between \ac{BS} and \ac{UE} and provides a \ac{NLoS} path non-resolvable from the \ac{LoS}. These two cases and close locations to them aggravate the performance, while it is mostly of a high quality. 

\section{Conclusions}\label{section:conclusion}

In this paper, we considered a single \ac{BS} transmitting a \ac{mmWave} \ac{OFDM} signal and a multi-antenna \ac{UE} receiving the \ac{LoS} path and at least one resolvable \ac{NLoS} path. The objective was to solve the snapshot 6D localization problem whereby the 3D position, the 3D orientation, and the clock bias of the \ac{UE} have to be estimated, as well as the positions of the incidence points. Two estimation routines, namely \ac{ML} estimation (which is a high-dimensional non-convex optimization problem over a product of manifolds) and an ad-hoc routine, were applied and their performance was evaluated. Although the performance of the former attains the \ac{CRB}, the latter provides estimates based on geometrical arguments. These estimates closely follow the \ac{CRB} for a large transmit power range, and serve as initialization to the recursive algorithms for solving the \ac{ML} complex estimation problem. In contrast, the proposed ad-hoc solution reduces the complexity to a single 1-dimensional search over a finite interval, combined with closed-form expressions. 
After obtaining an efficient \ac{ML}-based estimator, %in assistance with the \ac{ML} able to achieve the lower bounds, clarified the path forward, to
we evaluated the impact of different parameters, such as bandwidth, number of antennas, number of \ac{NLoS} paths, etc., through evaluation of \ac{CRB}. These results indicate that at least 2 \ac{NLoS} paths are needed to render the problem identifiable for most geometric configurations.  

% \ColorBlue{Nevertheless, there are several  limitations to the current work, which should be addressed in future studies. First of all, we did not provide a solution to the obstructed \ac{LoS} scenarios. Secondly, we only considered single-bounce \ac{NLoS} paths in this work, while paths might be reflected from more than one incidence point, which adversely affects the localization performance. }

\ColorBlue{There are several possible extensions of the current work, which may be addressed in future studies. First of all, a solution to the obstructed \ac{LoS} scenarios should be developed. Secondly, besides single-bounce \ac{NLoS} paths, there may double- and multiple-bounce reflections, and the performance of localization in these conditions should be investigated.}

%The reason was that multi-bounce \ac{NLoS} paths adversely affect the localization accuracy and should be excluded from the procedure. However, proposing a more general framework with a routine to identify multi-bounce \ac{NLoS} paths is postponed to future research.}

% Although we considered a general framework, there are still several scenarios to be explored for future research, e.g., where the \ac{LoS} is blocked, or under the presence of multi-bounce \ac{NLoS} paths. 

%, each originating from several scattering points or reflection surfaces, are considered.%One could also consider co-existence of intelligent surfaces as additional steerable signal sources.

\appendices

\section{Partial Derivatives of \acp{AoA}, \acp{AoD}, and \acp{ToA} w.r.t. Localization Parameters} \label{appendix: Partial Derivatives of AoAs, AoDs, and ToAs}
\subsection{Auxiliary Variables}
We define auxiliary variables
\begin{subequations}
\begin{align}
    \mathbf{u}_{\A,m\neq0} &= \frac{\p_m-\p_\UE}{\Vert\p_m-\p_\UE\Vert}, \quad \mathbf{u}_{\A,0} = \frac{\p_\BS-\p_\UE}{\Vert\p_\BS-\p_\UE\Vert},\\
    \mathbf{u}_{\D,m\neq0} &= \frac{\p_m-\p_\BS}{\Vert\p_m-\p_\BS\Vert}, \quad
    \mathbf{u}_{\D,0} = \frac{\p_\UE-\p_\BS}{\Vert\p_\UE-\p_\BS\Vert},
\end{align}
\end{subequations}
as well as $\mathbf{u}_1 = [1,0,0]\Transpose$, $\mathbf{u}_2 = [0,1,0]\Transpose$, and $\mathbf{u}_3 = [0,0,1]\Transpose$, for later use. Considering $\R_\UE=[\vecR_{\UE,1},\vecR_{\UE,2},\vecR_{\UE,3}]$ and $\R_\BS=[\vecR_{\BS,1},\vecR_{\BS,2},\vecR_{\BS,3}]$ gives $\R_\UE\mathbf{u}_i=\vecR_{\UE,i}$ and $\R_\BS\mathbf{u}_i=\vecR_{\BS,i}$, $i=1,2,3$.

\subsection{Mathematical Identities}
The following mathematical identities are used in obtaining the derivatives: 
\begin{subequations}
\begin{align}
    &\frac{\partial}{\partial \mathbf{X}} \mathbf{a}\Transpose \mathbf{X}\Transpose \mathbf{b}= \mathbf{b} \mathbf{a}\Transpose, \label{identitiy: derivative wrt matrix}\\
    &\frac{\partial}{\partial \mathbf{x}} \mathbf{a}\Transpose \mathbf{x}= \mathbf{a}, \label{identitiy: derivative wrt vector}\\
    &\frac{\partial}{\partial \mathbf{x}} \acos (v(\mathbf{x})) = -\frac{1}{\sqrt{1-v^2(\mathbf{x})}}  \frac{\partial v(\mathbf{x})}{\partial\mathbf{x}}, \label{identity: derivative of acos}\\
    &\frac{\partial}{\partial \mathbf{x}} \atan 2({v}(\mathbf{x}),{w}(\mathbf{x})) = \frac{ {w}(\mathbf{x})\frac{\partial v(\mathbf{x})}{\partial\mathbf{x}}-{v}(\mathbf{x})\frac{\partial w(\mathbf{x})}{\partial\mathbf{x}}}{{v}^2(\mathbf{x})+{w}^2(\mathbf{x})},\label{identity: derivative of atan}\\
    &\frac{\partial}{\partial \mathbf{x}} \frac{\mathbf{x} - \mathbf{a}}{\Vert \mathbf{x} - \mathbf{a} \Vert} = \frac{\mathbf{I}}{\Vert \mathbf{x} - \mathbf{a} \Vert} - \frac{(\mathbf{x} - \mathbf{a})(\mathbf{x} - \mathbf{a})\Transpose}{\Vert \mathbf{x} - \mathbf{a} \Vert^3}, \label{identity: derivative of unit norm vectors}\\
    &\frac{\partial}{\partial \mathbf{x}} \Vert \mathbf{x} - \mathbf{a} \Vert =  \frac{\mathbf{x} - \mathbf{a}}{\Vert \mathbf{x} - \mathbf{a} \Vert}.\label{identity: derivative of norm}
\end{align}
\end{subequations}
In \eqref{identity: derivative of acos} and \eqref{identity: derivative of atan}, if the derivatives are taken with respect to a matrix, we replace  $\mathbf{x}$ by $\mathbf{X}$, and the equations still hold. 
\subsection{Reformulation of AoAs and AoDs}
Using the defined auxiliary variables, one can express 
\begin{subequations} 
\begin{alignat}{12}
    &\bm{\theta}_{\A,m} &= [&\atan 2(&\vecR_{\UE,2}\Transpose &\mathbf{u}_{\A,m} &, &\vecR_{\UE,1}\Transpose &\mathbf{u}_{\A,m}) &, &\acos (&\vecR_{\UE,3}\Transpose &\mathbf{u}_{\A,m})]\Transpose,\notag\\
    &\bm{\theta}_{\D,m} &= [&\atan 2(&\vecR_{\BS,2}\Transpose &\mathbf{u}_{\D,m} &, &\vecR_{\BS,1}\Transpose &\mathbf{u}_{\D,m})  &, &\acos (&\vecR_{\BS,3}\Transpose &\mathbf{u}_{\D,m})]\Transpose.\notag
\end{alignat}
\end{subequations}

% \begin{subequations} 
% \begin{align}
%     \theta_{\A,m}^{\el} &= \acos (\vecR_{\UE,3}\Transpose \mathbf{u}_{\A,m}),\\
%     \theta_{\A,m}^{\az} &= \atan 2(\vecR_{\UE,2}\Transpose \mathbf{u}_{\A,m},\vecR_{\UE,1}\Transpose \mathbf{u}_{\A,m}),\\
%     \theta_{\D,m}^{\el} &= \acos (\vecR_{\BS,3}\Transpose \mathbf{u}_{\D,m}),\\
%     \theta_{\D,m}^{\az} &= \atan 2(\vecR_{\BS,2}\Transpose \mathbf{u}_{\D,m},\vecR_{\BS,1}\Transpose \mathbf{u}_{\D,m}).
% \end{align}
% \end{subequations}
\subsection{Derivatives with respect to UE Rotation Matrix}
We make use of \eqref{identitiy: derivative wrt matrix}, \eqref{identity: derivative of acos}, and \eqref{identity: derivative of atan}, 
\begin{subequations}
\begin{align}
    \frac{\partial \theta_{\A,m}^{\el}}{\partial \R_\UE} &= - 
    \frac{\mathbf{u}_{\A,m}\mathbf{u}_3\Transpose}{\sqrt{1-(\vecR_{\UE,3}\Transpose \mathbf{u}_{\A,m})^2}},\\
    %*********************************
    \frac{\partial \theta_{\A,m}^{\az}}{\partial \R_\UE} &= 
    \frac{(\vecR_{\UE,1}\Transpose \mathbf{u}_{\A,m})\mathbf{u}_{\A,m}\mathbf{u}_2\Transpose - (\vecR_{\UE,2}\Transpose \mathbf{u}_{\A,m})\mathbf{u}_{\A,m}\mathbf{u}_1\Transpose}{(\vecR_{\UE,1}\Transpose \mathbf{u}_{\A,m})^2 + (\vecR_{\UE,2}\Transpose \mathbf{u}_{\A,m})^2}.
\end{align}
\end{subequations}
Note that \acp{AoD} and \acp{ToA} have no dependence on $\R_\UE$, leading to partial derivative $\mathbf{0}_{3 \times 3}$. 
% \begin{subequations}
% \begin{align}
%     {\partial \theta_{\D,m}^{\el}}/{\partial \R_\UE} &= \mathbf{0}_{3 \times 3},\\
%     {\partial \theta_{\D,m}^{\az}}/{\partial \R_\UE} &= \mathbf{0}_{3 \times 3},\\
%     {\partial \tau_{m}}/{\partial \R_\UE} &= \mathbf{0}_{3 \times 3}.
% \end{align}
% \end{subequations}
\subsection{Derivatives with respect to \ac{UE} Position}
The derivatives with respect to $\p_\UE$ are obtained using the chain rule. We make use of \eqref{identitiy: derivative wrt vector}, \eqref{identity: derivative of acos} and \eqref{identity: derivative of atan} to obtain 
% \begin{subequations}
% \begin{align}
%     \frac{\partial \theta_{\A,m}^{\el}}{\partial \p_\UE} &= \frac{\partial \mathbf{u}_{\A,m}}{\partial \p_\UE}\frac{\partial \theta_{\A,m}^{\el}}{\partial \mathbf{u}_{\A,m}},\\
%     \frac{\partial \theta_{\A,m}^{\az}}{\partial \p_\UE} &= \frac{\partial \mathbf{u}_{\A,m}}{\partial \p_\UE}\frac{\partial \theta_{\A,m}^{\az}}{\partial \mathbf{u}_{\A,m}},\\
%     \frac{\partial \theta_{\D,m}^{\el}}{\partial \p_\UE} &= \frac{\partial \mathbf{u}_{\D,m}}{\partial \p_\UE}\frac{\partial \theta_{\D,m}^{\el}}{\partial \mathbf{u}_{\D,m}},\\
%     \frac{\partial \theta_{\D,m}^{\az}}{\partial \p_\UE} &= \frac{\partial \mathbf{u}_{\D,m}}{\partial \p_\UE}\frac{\partial \theta_{\D,m}^{\az}}{\partial \mathbf{u}_{\D,m}},
% \end{align}
% \end{subequations}
% for $\forall m$, in which
\begin{subequations}
\begin{align}
    \frac{\partial \theta_{\A,m}^{\el}}{\partial \mathbf{u}_{\A,m}} &= -\frac{\vecR_{\UE,3}}{\sqrt{1-(\vecR_{\UE,3}\Transpose \mathbf{u}_{\A,m})^2}},\\
    %*********************************
    \frac{\partial \theta_{\A,m}^{\az}}{\partial \mathbf{u}_{\A,m}} &= 
    \frac{(\vecR_{\UE,1}\Transpose \mathbf{u}_{\A,m})\vecR_{\UE,2} - (\vecR_{\UE,2}\Transpose \mathbf{u}_{\A,m})\vecR_{\UE,1}}{(\vecR_{\UE,1}\Transpose \mathbf{u}_{\A,m})^2 + (\vecR_{\UE,2}\Transpose \mathbf{u}_{\A,m})^2}, \\
    %*********************************
    \frac{\partial \theta_{\D,m}^{\el}}{\partial \mathbf{u}_{\D,m}} &= -\frac{\vecR_{\BS,3}}{\sqrt{1-(\vecR_{\BS,3}\Transpose \mathbf{u}_{\D,m})^2}},\\
    %*********************************
    \frac{\partial \theta_{\D,m}^{\az}}{\partial \mathbf{u}_{\D,m}} &=
    \frac{(\vecR_{\BS,1}\Transpose \mathbf{u}_{\D,m})\vecR_{\BS,2} - (\vecR_{\BS,2}\Transpose \mathbf{u}_{\D,m})\vecR_{\BS,1}}{(\vecR_{\BS,1}\Transpose \mathbf{u}_{\D,m})^2 + (\vecR_{\BS,2}\Transpose \mathbf{u}_{\D,m})^2}, 
\end{align}
\end{subequations}
and \eqref{identity: derivative of unit norm vectors} to obtain
\begin{subequations}
\begin{align}
    {\partial \mathbf{u}_{\A,m\neq0}}/{\partial \p_\UE} &= \left(\mathbf{u}_{\A,m} \mathbf{u}_{\A,m}\Transpose -\mathbf{I}_3\right)/{\Vert \p_m - \p_\UE \Vert},\\
    {\partial \mathbf{u}_{\A,0}}/{\partial \p_\UE} &= \left(\mathbf{u}_{\A,0} \mathbf{u}_{\A,0}\Transpose - \mathbf{I}_3 \right)/{\Vert \p_\BS - \p_\UE \Vert},\\
    {\partial \mathbf{u}_{\D,0}}/{\partial \p_\UE} &= \left(\mathbf{I}_3 - \mathbf{u}_{\D,0} \mathbf{u}_{\D,0}\Transpose \right)/{\Vert \p_\UE - \p_\BS \Vert},
\end{align}
\end{subequations}
and ${\partial \mathbf{u}_{\D,m\neq0}}/{\partial \p_\UE} = \mathbf{0}_{3\times3}$. Also, considering \eqref{identity: derivative of norm} gives
\begin{align} \label{eq: derivative of ToAs wrt p_UE}
    \frac{\partial \tau_m}{\partial \p_\UE} = 
    \begin{cases}
        (\p_\UE-\p_\BS)/(c~\Vert\p_\UE-\p_\BS\Vert) & m=0\\
        (\p_\UE-\p_m)/(c~\Vert\p_\UE-\p_m\Vert) & m \neq 0
    \end{cases}.
\end{align}

\subsection{Derivatives with respect to Incidence Points Positions}
The derivatives with respect to \ac{IP} positions are also obtained using the chain rule. We note that ${\partial \mathbf{u}_{\A,m}}/{\partial \p_n}=\mathbf{0}_3,~n \neq m$, for $m=0,\ldots,M$ and $n=1,\ldots,M$, with the same case for ${\partial \mathbf{u}_{\A,m}}/{\partial \p_n}$ and ${\partial \tau_m}/{\partial \p_n}$, while for $m \neq 0$
\begin{subequations}
\begin{align}
    {\partial \mathbf{u}_{\A,m}}/{\partial \p_m} &= \left(\mathbf{I}_3 - \mathbf{u}_{\A,m} \mathbf{u}_{\A,m}\Transpose \right)/\Vert \p_m - \p_\UE \Vert, \\
    {\partial \mathbf{u}_{\D,m}}/{\partial \p_m} &= \left(\mathbf{I}_3 - \mathbf{u}_{\D,m} \mathbf{u}_{\D,m}\Transpose \right)/\Vert \p_m - \p_\BS \Vert \\
    {\partial \tau_m}/{\partial \p_m} &= (\p_m-\p_\BS)/(c~\Vert\p_m-\p_\BS\Vert) \nonumber \\ &+ (\p_m-\p_\UE)/(c~ \Vert\p_m-\p_\UE\Vert).
\end{align}
\end{subequations}

\subsection{Derivatives with respect to Clock Bias}
The angles have no dependence on $b$, while ${\partial \tau_m}/{\partial b}= 1$.

%\section{Proof of Proposition \ref{proposition: characterization of rotation matrix}} \label{proof:characterization of rotation matrix}
%Clearly, the unit-norm \ac{LoS} arrival direction $\mathbf{d}_{\A,0}$ is left invariant by the rotation matrix $\mathbf{Q}_{\mathbf{d}_{\A,0}}(\psi),~\forall \psi \in [0,2\pi)$, i.e., 
%\begin{align} \label{eq:LoS arrival direction is left invariant}
 %   \mathbf{d}_{\A,0} = \mathbf{Q}_{\mathbf{d}_{\A,0}}(\psi) \mathbf{d}_{\A,0},~\forall \psi \in [0,2\pi).
%\end{align}
%Let us re-write the equation \eqref{eq:AoD-AoA LoS condition}, using $\widetilde{\R}$ as
%\begin{align} \label{eq:AoD-AoA LoS condition using R_tilde}
 %   \R_\BS\mathbf{d}_{\D,0} = - \widetilde{\R}\mathbf{d}_{\A,0}.
%\end{align}
%Substituting \eqref{eq:LoS arrival direction is left invariant} in \eqref{eq:AoD-AoA LoS condition using R_tilde} yields:
%\begin{align} \label{eq:AoD-AoA LoS condition using R_tilde and Q}
 %   \R_\BS\mathbf{d}_{\D,0} = - \widetilde{\R} \mathbf{Q}_{\mathbf{d}_{\A,0}}(\psi) \mathbf{d}_{\A,0},~\forall \psi \in [0,2\pi].
%\end{align}
%Comparing \eqref{eq:AoD-AoA LoS condition using R_tilde and Q} with \eqref{eq:AoD-AoA LoS condition} characterizes the rotation matrix as
%\begin{align}
 %   \R(\psi) = \widetilde{\R}\mathbf{Q}_{\mathbf{d}_{\A,0}}(\psi).
%\end{align}
\section{Solving the Optimization Problem for the Shortest Distance Between Half-lines} \label{appendix: solving opt for shortest distance}

For the half-lines $\ell_1=\{\p \in \mathbb{R}^3:\p = \p_1+t_1\mathbf{d}_1,t_1\ge0\}$ and $\ell_2=\{\p \in \mathbb{R}^3:\p = \p_2+t_2\mathbf{d}_2,t_2\ge0\}$, the shortest distance $\delta_{\mathrm{min}}$ is obtained from $\delta_{\min}^2 = \min_{\boldsymbol{t}=[t_1,t_2]\Transpose} \Vert (\p_1+t_1\mathbf{d}_1) - (\p_2+t_2\mathbf{d}_2) \Vert^2,~\text{s.t.}~\boldsymbol{t} \ge \mathbf{0}_2$, which is a quadratic convex optimization problem in $\boldsymbol{t}$, and its solution is found by writing the K.K.T.~conditions \cite{boyd2004convex}. We utilize a simpler procedure in which we first obtain the unconstrained optimal solution  
\begin{subequations}\label{eq: optimal solution}
\begin{alignat}{2}
    t_1^* &= -&\mathbf{d}_1\Transpose\left(\mathbf{I} - \mathbf{d}_2 \mathbf{d}_2\Transpose \right)\p_{12}/ \left(1-(\mathbf{d}_1\Transpose\mathbf{d}_2)^2\right),\\
    t_2^* &= &\mathbf{d}_2\Transpose\left(\mathbf{I} - \mathbf{d}_1 \mathbf{d}_1\Transpose \right)\p_{12}/ \left(1-(\mathbf{d}_1\Transpose\mathbf{d}_2)^2\right),
\end{alignat}
\end{subequations}
with $\p_{12}\triangleq\p_1-\p_2$. If $\boldsymbol{t}^* > \mathbf{0}_2$, the solution is $\delta_{\min}= |\mathbf{n}\Transpose\mathbf{p}_{12}|$ where $\mathbf{n} = (\mathbf{d}_1\times\mathbf{d}_2)/( \Vert\mathbf{d}_1\times\mathbf{d}_2\Vert)$.
\begin{proof}
Substituting the optimal solution \eqref{eq: optimal solution} in $\Vert (\p_1+t_1\mathbf{d}_1) - (\p_2+t_2\mathbf{d}_2) \Vert$ gives $\delta_{\min} = \Vert \mathbf{P}_{\perp}(\mathbf{D}) \cdot \mathbf{p}_{12} \Vert$, where $\mathbf{D}\triangleq[\mathbf{d}_1,\mathbf{d}_2]$, and $\mathbf{P}_{\perp}(\mathbf{D}) \triangleq \mathbf{I}-\mathbf{D}(\mathbf{D}\Transpose\mathbf{D})^{-1}\mathbf{D}\Transpose = \mathbf{I} - \big( \mathbf{d}_1 \mathbf{d}_1\Transpose\left(\mathbf{I} - \mathbf{d}_2 \mathbf{d}_2\Transpose \right) + \mathbf{d}_2 \mathbf{d}_2\Transpose\left(\mathbf{I} - \mathbf{d}_1 \mathbf{d}_1\Transpose\right) \big)/\big( {1-(\mathbf{d}_1\Transpose\mathbf{d}_2)^2}\big)$ is the projector onto the subspace orthogonal to the one spanned by $\mathbf{d}_1$ and $\mathbf{d}_2$, which is in turn spanned by the unit-norm vector $\mathbf{n}$ normal to $\mathbf{d}_1$ and $\mathbf{d}_2$, given by $\mathbf{n} = (\mathbf{d}_1\times\mathbf{d}_2)/( \Vert\mathbf{d}_1\times\mathbf{d}_2\Vert)$. Hence, $\delta_{\min} = \Vert(\mathbf{n}\Transpose \mathbf{p}_{12})~\mathbf{n}\Vert = |\mathbf{n}\Transpose\mathbf{p}_{12}|$.
\end{proof}
\noindent Otherwise, we obtain 
$\check{t}_1 =-\mathbf{d}_1\Transpose\p_{12}$, 
$\check{t}_2 = \mathbf{d}_2\Transpose\p_{12}$, 
$\check{\lambda}_1 = \mathbf{d}_1\Transpose\mathbf{P}_{\perp}(\mathbf{d}_2)\p_{12}$, and 
$\check{\lambda}_2 =-\mathbf{d}_2\Transpose\mathbf{P}_{\perp}(\mathbf{d}_1)\p_{12}$,
where $\mathbf{P}_{\perp}(\mathbf{d}) \triangleq \mathbf{I} - \mathbf{d} \mathbf{d}\Transpose$. 
If $[\check{t}_1 ,\check{\lambda}_2]\Transpose>\mathbf{0}_2$, then $\delta_{\min}=\check{\delta}_1$, and if $[\check{t}_2 ,\check{\lambda}_1]\Transpose>\mathbf{0}_2$, then $\delta_{\min}=\check{\delta}_2$, where
$\check{\delta}_1 \triangleq (\mathbf{p}_{12}\Transpose\mathbf{P}_{\perp}(\mathbf{d}_1)\p_{12})^{1/2}$ and
$\check{\delta}_2 \triangleq (\mathbf{p}_{12}\Transpose\mathbf{P}_{\perp}(\mathbf{d}_1)\p_{12})^{1/2}$
are both non-negative, due to the Cauchy–Schwarz inequality. Otherwise, $\delta_{\min}= \Vert\p_{12}\Vert$. Obtaining expressions is straightforward. As the unconstrained solution $\boldsymbol{t}^*$ often satisfies the constraints, especially in high \ac{SNR} regimes, this approach is more efficient.

\section{Proof of the Closed-From Expression for the Closest Point to Skew Lines} \label{proof: nearest point to skew lines}
We determine $\p_0$ to be mutually closest to the half-lines $\ell_1=\{\p \in \mathbb{R}^3:\p = \p_1+t_1\mathbf{d}_1,t_1\ge0\}$ and $\ell_2=\{\p \in \mathbb{R}^3:\p = \p_2+t_2\mathbf{d}_2,t_2\ge0\}$, in a least-squares sense, so that $\mathrm{d}^2(\p_0,\ell_1)+\mathrm{d}^2(\p_0,\ell_2)$, with  $\mathrm{d}(\p_0,\ell_i)$ denoting the distance of $\p_0$ to $\ell_i,~i=1,2$, is minimized. According to Pythagorean theorem, $\mathrm{d}^2(\p_0,\ell_i) = \Vert\p_0-\p_i\Vert^2-\left( (\p_0-\p_i)\Transpose\mathbf{d}_i\right)^2$, where $(\p_0-\p_i)\Transpose\mathbf{d}_i$ is the projection of $(\p_0-\p_i)$ on line $\ell_i$. Taking gradient of $\mathrm{d}^2(\p_0,\ell_1)+\mathrm{d}^2(\p_0,\ell_2)$ with respect to $\p_0$ and setting it to $\mathbf{0}$ results in $\p_0 = \mathbf{A}^{-1} \mathbf{b}$, where $\mathbf{A}=(\mathbf{I}-\mathbf{d}_1\mathbf{d}_1\Transpose)+(\mathbf{I}-\mathbf{d}_2\mathbf{d}_2\Transpose)$ and $\mathbf{b}=(\mathbf{I}-\mathbf{d}_1\mathbf{d}_1\Transpose)\p_1+(\mathbf{I}-\mathbf{d}_2\mathbf{d}_2\Transpose)\p_2$. Setting $\p_1=\p_\BS$, $\p_2=\p_\UE(1)$, $\mathbf{d}_1=\R_\BS\mathbf{d}_{\D,m}$, and $\mathbf{d}_2=\hat{\R}_\UE\mathbf{d}_{\A,m}$ gives \eqref{eq: estimate of nearest points}.

\balance 
\bibliographystyle{IEEEtran}
\bibliography{IEEEabrv,references} 
\end{document}